\documentclass[11pt,a4paper,twoside,openright]{book}
\usepackage[T1]{fontenc}
\usepackage[utf8]{inputenc}
\usepackage{lmodern}
\usepackage[a4paper,textwidth=150mm,textheight=230mm,centering,headheight=14pt,headsep=7mm,footskip=12mm]{geometry}
\usepackage{microtype}
\usepackage{amsmath,amssymb,amsthm}
\usepackage{tikz-cd}
\usetikzlibrary{decorations.pathmorphing}

\usepackage{enumitem,needspace,etoolbox}
\usepackage{titlesec,tocloft,fancyhdr,emptypage}
\usepackage[numbers,sort&compress]{natbib}
\usepackage{xurl}
\usepackage[unicode,hidelinks,bookmarksnumbered,linktoc=all]{hyperref}
\hypersetup{pdftitle={Freely generated categorical structures and automatic differentiation},pdfauthor={Fernando Lucatelli Nunes},pdfsubject={Introduction, conclusion and summaries of the doctoral thesis, Utrecht University, 2026},pdfkeywords={categorical semantics, automatic differentiation, CHAD, Grothendieck constructions}}
\titleformat{\chapter}[display]{\normalfont}{\large\scshape\chaptertitlename\ \thechapter}{12pt}{\Huge\bfseries\raggedright\hyphenpenalty=10000}
\titlespacing*{\chapter}{0pt}{24pt}{30pt}
\titleformat{\section}{\large\bfseries}{\thesection}{0.8em}{}
\titleformat{\subsection}{\normalsize\bfseries}{\thesubsection}{0.8em}{}
\makeatletter
\renewcommand{\chaptermark}[1]{\if@mainmatter\markboth{\thechapter.\ #1}{\thechapter.\ #1}\else\markboth{#1}{#1}\fi}
\makeatother
\renewcommand{\sectionmark}[1]{}
\fancypagestyle{plain}{\fancyhf{}\fancyfoot[LE,RO]{\thepage}}
\newcommand{\op}{\mathsf{op}}
\newcommand{\Vect}{\mathsf{Vect}}
\newcommand{\catM}{\mathsf{Man}}
\newcommand{\catSet}{\mathsf{Set}}
\newcommand{\catTop}{\mathsf{Top}}
\newcommand{\catC}{\mathsf{C}}
\newcommand{\Fam}[1]{\mathsf{Fam}\left(#1\right)}
\newtheorem{theorem}{Theorem}[section]
\newtheorem{lemma}[theorem]{Lemma}

\theoremstyle{remark}

\AtBeginEnvironment{lemma}{\Needspace{8\baselineskip}}

\begin{document}
\frontmatter
\begin{titlepage}
\centering
\vspace*{24mm}
{\Huge\bfseries Freely generated categorical structures and automatic differentiation\par}
\vspace{20mm}
{\Large Fernando Lucatelli Nunes\par}
\vspace{18mm}
{\large Doctoral thesis\par Utrecht University\par}
\vspace{12mm}
{\large Introduction, conclusion and summaries\par}
\vfill
{\large 2026\par}
\vspace{8mm}
{arXiv version\par}
\end{titlepage}

\thispagestyle{empty}
\null\vfill
\noindent Department of Information and Computing Sciences\\
Utrecht University
\par\bigskip
\noindent This thesis was completed with financial support from the NWO Veni grant VI.Veni.201.124 and the ERC project \textsc{FoRECAST}.
\par\bigskip
\noindent Dit proefschrift is tot stand gekomen met financi\"ele steun van de NWO Veni-beurs VI.Veni.201.124 en het ERC-project \textsc{FoRECAST}.
\par\bigskip
\noindent Copyright \textcopyright\ 2026 Fernando Lucatelli Nunes.
\cleardoublepage

\chapter*{About this version}
\phantomsection
\addcontentsline{toc}{chapter}{About this version}
\markboth{About this version}{About this version}
This version contains the introduction, conclusion, and Dutch and English summaries of the doctoral thesis \emph{Freely generated categorical structures and automatic differentiation}. The six research papers forming Chapters~2--7 of the full thesis are not reproduced here; they are available separately at the links below. The original chapter numbering is retained: the introduction is Chapter~1 and the conclusion is Chapter~8. Throughout this version, references to chapters and descriptions of the thesis structure refer to the full thesis. The bibliography contains the works cited in this version.

Chapters~2, 3, 5 and~7 are joint work with Matthijs V\'ak\'ar; Chapter~4 is joint work with Gordon Plotkin and Matthijs V\'ak\'ar; and Chapter~6 is joint work with Rui Prezado and Matthijs V\'ak\'ar. In the joint papers forming Chapters~2--7 of the full thesis, I took a leading role in the mathematical development and preparation of the manuscripts, with particular responsibility for the formulation and proof of the main results. The work was developed in close collaboration with the coauthors named above, through sustained discussion and refinement of the ideas, arguments and exposition.

\section*{Papers in the full thesis}
\begin{enumerate}[start=2,leftmargin=2em,itemsep=7pt,parsep=0pt,topsep=6pt]
\item \emph{Automatic differentiation for ML-family languages: Correctness via logical relations}~\cite{nunes_vakar_2024_ml}.\\
\href{https://arxiv.org/abs/2210.07724v3}{arXiv:2210.07724v3}.
\item \emph{CHAD for expressive total languages}~\cite{nunes_vakar_2021_chad}.\\
\href{https://arxiv.org/abs/2110.00446v2}{arXiv:2110.00446v2}.
\item \emph{Unraveling the iterative CHAD}~\cite{nunes_iterative_chad_2025}.\\
\href{https://arxiv.org/abs/2505.15002v3}{arXiv:2505.15002v3}.
\item \emph{Free Doubly-Infinitary Distributive Categories are Cartesian Closed}~\cite{nunes_vakar_2024_doubly_inf_distrib}.\\
\href{https://doi.org/10.1007/s10485-026-09887-7}{Published article}; \href{https://arxiv.org/abs/2403.10447v10}{arXiv:2403.10447v10}.
\item \emph{Free extensivity via distributivity}~\cite{nunes_extensivity_2025}.\\
\href{https://doi.org/10.4171/PM/2129}{Published article}; \href{https://arxiv.org/abs/2405.02185}{arXiv:2405.02185}.
\item \emph{Monoidal closure of Grothendieck constructions via $\Sigma$-tractable monoidal structures and Dialectica formulas}~\cite{nunes_monoidal_grothendieck_2024}.\\
\href{https://www.tac.mta.ca/tac/volumes/44/35/44-35abs.html}{Published article}; \href{https://arxiv.org/abs/2405.07724v4}{arXiv:2405.07724v4}.
\end{enumerate}
\cleardoublepage
\tableofcontents
\cleardoublepage
\mainmatter
\chapter{Introduction}

This thesis rests on the conviction that theory and practice in computer science should inform one another. As Donald Knuth emphasised in \emph{Theory and Practice}~\cite{knuth1991theorypractice}, theoretical and practical work gain from sustained interaction. Semantic foundations can clarify what an implementation must achieve and provide methods for establishing its correctness. Conversely, the demands of implementation can reveal questions that are not apparent from an abstract formulation alone. The work presented here develops this interaction through automatic differentiation and categorical semantics.

In programming languages, this principle takes a concrete form: practical language features guide the theoretical questions, while semantic frameworks provide a basis for specifying transformations and proving them correct. The resulting constructions may, in turn, suggest improvements to the language or its implementation.

Differentiable programming \cite{baydin2018automatic,griewank2008evaluating} exemplifies this interaction. It plays a substantial role in scientific computing, machine learning, and optimisation. Its central idea is to compute derivatives of the functions denoted by programs, making differentiation available as part of programming itself. At the heart of this approach lies a systematic program transformation implementing the chain rule of calculus: \emph{automatic differentiation} (AD). AD has a long history, from its early development in numerical analysis \cite{griewank2008evaluating} to its reformulation in functional settings and more recent semantic treatments \cite{elliott2018essence,PLOTKU}.

In this thesis, automatic differentiation provides the concrete problem through which we develop and evaluate semantic frameworks for programming languages. We seek a \emph{compile-time} AD transformation for expressive languages, together with a precise semantics and a proof that the generated program computes the derivative of the original. Universal properties determine the transformation from its action on primitives; correctness additionally requires correct primitive derivatives and a semantic argument relating the transformed program to differentiation. This separation makes the construction modular and provides a foundation for subsequent work on optimisation, efficiency, and parallelism.

\section{Reverse-mode Automatic Differentiation}

Automatic differentiation raises a semantic question as well as an algorithmic one: how should a language represent derivatives, and how should differentiation interact with its constructs? Reverse mode makes these questions particularly apparent.

\emph{Reverse-mode} automatic differentiation transforms a program so that, in addition to its original outputs, it computes \emph{co-derivatives}, namely the action of the transpose of the derivative on cotangent vectors. For \(f : \mathbb{R}^m \to \mathbb{R}^n\), a reverse pass computes $Df(x)^t w$ for an output cotangent $w$, whereas a forward pass computes $Df(x)v$ for an input tangent $v$. Recovering the full Jacobian requires $n$ such reverse passes or $m$ forward passes. Reverse mode is therefore particularly attractive when there are many inputs and few outputs, as in the scalar loss functions of machine learning; forward mode is often preferable in the opposite regime. Actual costs also depend on the implementation and its storage requirements \cite{baydin2018automatic,griewank2008evaluating}.

At the same time, reverse-mode AD exposes a persistent gap between practical implementations and semantic foundations. Achieving a correct and efficient source-to-source transformation for expressive higher-order languages requires threading derivative computations through programs in a way that preserves their structure and intent. Widely used frameworks such as TensorFlow \cite{tensorflow2015}, PyTorch \cite{paszke2019pytorch}, and JAX \cite{jax2018github} combine graph representations, tracing, compilation, and runtime support in different ways. The question pursued here concerns the language-level foundations: how to derive a compositional source transformation and prove its correctness for the constructs of an expressive typed language.

This difficulty leads to the central question of the thesis: \emph{can both forward- and reverse-mode AD be formulated as compositional compile-time transformations with correctness proofs} for expressive programming languages supporting inductive types, recursion, non-termination, and other advanced features? We argue that \emph{categorical semantics} provides a suitable framework. Our approach builds on the categorical treatment of forward-mode AD via dual numbers for the simply typed $\lambda$-calculus \cite{DBLP:conf/fossacs/HuotSV20}, Elliott's categorical account of AD~\cite{elliott2018essence}, and V\'ak\'ar's reverse-mode transformation for the simply typed $\lambda$-calculus~\cite{DBLP:conf/esop/Vakar21}, known as \emph{Combinatory Homomorphic Automatic Differentiation} (CHAD).

\section{Categorical Semantics of Programming Languages}

Category theory has long provided a rigorous mathematical foundation for the semantics of programming languages, offering a precise framework for their denotational interpretation. Seminal works such as Lambek and Scott’s introduction to higher-order categorical logic \cite{lambek1986introduction}, Crole’s \emph{Categories for Types} \cite{crole1993categories}, and Jacobs’ \emph{Categorical Logic and Type Theory} \cite{jacobs1999categorical} demonstrate how rich type theories and $\lambda$-calculi can be modelled by categorical structures including cartesian closed categories, indexed categories, and Freyd categories.

From the outset, categorical methods have profoundly influenced both the theory and practice of programming languages. The following examples highlight some of the landmark developments:  
\begin{itemize} 
	\item Moggi's use of \emph{monads} in programming-language semantics provided a unifying account of computational effects \cite{moggi1989computational,moggi1991notions}, including state, exceptions, nondeterminism, partiality, continuations, probabilistic choice, and I/O. Since then, monads have become a standard tool in semantic modelling and in functional programming, most notably through their adoption in Haskell. 
	\item Fiore and Plotkin’s axiomatic domain theory identified categorical conditions, such as algebraic compactness, that solve recursive domain equations and thereby deliver a uniform treatment of recursive types and domains \cite{fiore1996axiomatic}. 
	\item Plotkin and Power’s theory of \emph{algebraic effects} refined the monadic approach through algebraic theories; subsequent work on effect handlers further developed this modular treatment of computational effects.  
	\item Levy's \emph{Call-by-Push-Value} (CBPV) analyses call-by-name and call-by-value within a common calculus of values and computations, with adjunctions playing a central role in its semantics \cite{levy1999call,levy2006cbpv}.
\end{itemize}

In this tradition, the aim of this thesis is to develop categorical methods that inform the design of program transformations and support proofs of their correctness. The resulting account of AD also raises questions of wider interest in categorical semantics, particularly concerning the structure of free constructions and indexed models.

On the one hand, categorical insights yield structure-preserving AD transformations with transparent correctness proofs. On the other hand, the demands of real implementations expose the limitations of existing frameworks and stimulate the development of richer theories that remain accountable to practice. Throughout this thesis, AD thus serves simultaneously as a stress test and a catalyst: it presses semantic frameworks to their limits while also benefiting from the conceptual clarity and guarantees they provide.

\section{Functoriality}

We now cast automatic differentiation in categorical terms. This is natural, since the chain rule ensures that differentiation respects identities and composition, which is the hallmark of functoriality. A categorical viewpoint sharpens both the specification and the implementation of AD: it enforces modularity by decomposing programs along categorical structure, clarifies semantics by presenting AD as a structure-preserving functor, and supports efficiency by enabling principled reuse and optimisation. Accordingly, if differentiation is to function as a reliable program transformation, it should arise as a functor that preserves relevant categorical structure.

From now on, we work in the setting of differentiable manifolds and maps between them, the standard framework for studying differentiation in classical mathematics. For readers unfamiliar with manifolds, it is enough to think of Euclidean spaces and their open subsets, as typically encountered in an introductory or intermediate course in calculus. Since manifolds are seldom implemented directly in practice, these simpler settings are often sufficient for reasoning about modern differentiable programming. Nevertheless, manifolds also admit useful presentations by idempotents. Every finite-dimensional smooth manifold without boundary, assumed Hausdorff and second countable, can be embedded in a Euclidean space and realised as a smooth retract of an open neighbourhood. Such a presentation consists of an open set $U$ and a smooth idempotent $e:U\to U$. Maps between the images of two such idempotents are represented by smooth maps
\[
f:U\longrightarrow U',\qquad f=e'\circ f\circ e.
\]
We write $\catM$ for the category of finite-dimensional smooth manifolds and smooth maps, using the convention that each manifold has a fixed dimension. Accordingly, when using idempotent presentations we restrict to those whose images have a fixed dimension.

Assuming that $\catM$ models the phenomena of interest, a program transformation should be a structure-preserving functor that assigns to each differentiable map its (co-)derivative.  Since AD relies on the chain rule, it ought to behave functorially. More precisely, at its most basic, the co-derivative defines a functor
\begin{equation}\label{Differentiation-BASIC}
	\mathfrak{D}^t : \mathsf{1}\downarrow \catM \;\longrightarrow\; \Vect^{\op}
\end{equation}
where $\mathsf{1}$ denotes the one-point manifold and $\mathsf{1}\downarrow \catM$ is the coslice category of pointed manifolds and point-preserving smooth maps; see standard references on comma categories such as \cite[Ch.~II.6]{MacLane1998} or \cite[Sec.~2.3]{AdamekHerrlichStrecker2006}. Dually, the derivative is also a functor $\mathfrak{D} : \mathsf{1}\downarrow \catM \to \Vect$.

In practice, Euclidean spaces are implemented without base points. This motivates reframing differentiation as a functor
\begin{equation}\label{eq:derivative-functor-question}
	\mathsf{D} : \catM \;\longrightarrow\; \catC .
\end{equation}
The central question then becomes: what is the appropriate target category $\catC$? Two main paradigms, relying on distinct answers to this question, have been developed within the categorical semantics of AD: the dual numbers approach and CHAD; see Chapters~\textbf{2} and \textbf{3}, respectively.

As explained later, dual numbers offer a principled approach to forward-mode AD. They also admit reverse-mode formulations, studied in Chapter~2. CHAD makes a different choice of target structure, representing the primal value and its associated cotangent space separately; this makes the dependence of cotangents on primal values explicit.

In CHAD, the cotangent spaces over all points of a manifold are collected into a family. The functor \eqref{Differentiation-BASIC} thereby induces
\begin{equation}
	cha\mathfrak{D} : \catM \;\longrightarrow\; \mathsf{Fam}(\Vect^{\op}),
\end{equation}
where $\mathsf{Fam}(\Vect^{\op})$ is the free coproduct completion of $\Vect^{\op}$ (see \cite{AdamekFreeCoproductCompletion}). Concretely, $cha\mathfrak{D}$ maps a manifold $M$ to the pair $(M, TM)$, where $M$ denotes the underlying set of $M$ (by slight abuse of language), and $TM : M \to \Vect^{\op}$, with the indexing set viewed as a discrete category, assigns to each $x\in M$ the cotangent space at $x$. A morphism $f : M \to N$ in $\catM$ is mapped to the pair $(f,f')$, where $f$ is the underlying function and $f'$ is the family of linear maps
\[
\big(f'_x : TN(f(x)) \to TM(x)\big)_{x\in M},
\]
representing the respective co-derivatives.
Since $\mathsf{Fam}(\Vect^{\op})$ models the category of containers for suitable dependently typed languages with linear types (see Chapters~3 and~4; cf.\ containers \cite{AbbottAltenkirchGhaniFoSSaCS03}), this viewpoint underlies the CHAD construction introduced in~\cite{DBLP:conf/esop/Vakar21}. Chapters~3 and~4 develop this functorial account, while Chapter~2 treats dual numbers.

\subsection{Structure-preserving functoriality}

Program transformations that are \emph{structure-preserving} admit systematic methods of reasoning. They align with the universal properties that characterise the syntactic and semantic structures of programming languages (modulo $\beta\eta$-equivalence) and support proof techniques such as logical relations. These methods organise the correctness proof around the structure of the language and the behaviour of its primitives.

Returning to the functor \eqref{eq:derivative-functor-question}, it is not enough for differentiation to be merely functorial. Both practical and theoretical considerations point to a stronger requirement: the functor must also be \emph{structure-preserving}, acting as a morphism of the categorical structure we care about. This requirement is not incidental but reflects a general methodological principle that extends beyond differentiation to program transformations more broadly.

Differentiation is an especially suitable test case for this principle, since it is already essentially structure-preserving. The task, then, is to make this precise by establishing a functorial differentiation that preserves the relevant categorical structure. Concretely, we show:

\begin{theorem}
	The functor $cha\mathfrak{D} : \catM \longrightarrow \mathsf{Fam}(\Vect^{\op})$ preserves finite products.
\end{theorem}

\begin{proof}
The cotangent space of a product has the canonical decomposition
\[
T^*_{(x,y)}(M\times N)\cong T_x^*M\oplus T_y^*N.
\]
The direct sum is a coproduct in $\Vect$, hence a product in $\Vect^{\op}$. Together with the product of the indexing sets, this is precisely the product of families. The cotangent maps of the manifold projections are the corresponding direct-sum injections, and the cotangent space of the one-point manifold is zero, giving the terminal family.
\end{proof}

\subsection{Forward CHAD and dual numbers}

The relation between forward CHAD and dual-numbers AD is especially transparent at first-order types. Let $\Vect$ denote the category of real vector spaces and linear maps, let $U:\Vect\to\catSet$ be the underlying-set functor, and consider the total-elements functor
\[
Q:=\Sigma\circ\mathsf{Fam}(U):\mathsf{Fam}(\Vect)\longrightarrow\catSet,
\qquad Q(A,V)=\coprod_{a\in A}U(V_a).
\]
Here $\Sigma:\mathsf{Fam}(\catSet)\to\catSet$ takes the disjoint union of a family. Thus $Q$ retains both the primal value and a tangent vector; it is not the projection $(A,V)\mapsto A$. On a family morphism $(f,\varphi):(A,V)\to(B,W)$ it acts by
\[
Q(f,\varphi)(a,v)=\bigl(f(a),\varphi_a(v)\bigr).
\]
The products of families have fibres $V_a\times W_b$, while coproducts concatenate families. Consequently, $Q$ preserves finite products and coproducts, including their nullary cases: the terminal family is $(1,0)$ and its total set is a singleton.

For the simply typed total source language with products, variants, and function types, the forward CHAD interpretation gives a composite
\[
\mathsf{Syn}\xrightarrow{\;\overrightarrow{\mathsf{CHAD}}\;}
\mathsf{Fam}(\Vect)\xrightarrow{\;Q\;}\catSet.
\]
On a ground type $\mathbb{R}^n$, forward CHAD assigns the constant tangent family $(\mathbb{R}^n,\underline{\mathbb{R}^n})$. For a smooth primitive $f:\mathbb{R}^n\to\mathbb{R}^m$, the composite therefore gives
\[
(x,\dot x)\longmapsto\bigl(f(x),D f(x)\dot x\bigr).
\]
This is precisely the forward dual-numbers interpretation. Since both interpretations preserve finite products and coproducts and agree on the generators, they agree, up to the canonical comparison isomorphisms, on the first-order fragment generated by these operations:
\[
\begin{tikzcd}[column sep=large,row sep=large]
\mathsf{Syn}_{\mathrm{fo}}
  \arrow[r,"{\overrightarrow{\mathsf{CHAD}}}"]
  \arrow[dr,swap,"{\mathsf{D}_{\mathrm{dual}}}"]
& \mathsf{Fam}(\Vect)\arrow[d,"Q"]\\
& \catSet .
\end{tikzcd}
\]
At real inputs and outputs, the correctness theorems also identify their results for programs that use higher-order constructs internally: both compute the same primal result and directional derivative.

The distinction arises at function types. Although $\mathsf{Fam}(\Vect)$ is cartesian closed, $Q$ does not preserve exponentials. A small example makes the obstruction explicit. Put $X=(1,\mathbb{R})$. The exponential $X^X$ has one shape for each linear map $\mathbb{R}\to\mathbb{R}$, and a copy of $\mathbb{R}$ over each shape. Hence the canonical comparison
\[
Q(X^X)\longrightarrow Q(X)^{Q(X)}
\]
identifies a pair $(a,b)\in\mathbb{R}^2$ with the affine function $v\mapsto av+b$; it does not include, for example, $v\mapsto v^2$. The same obstruction occurs on the actual ground-type interpretation $A=(\mathbb{R},\underline{\mathbb{R}})$: the comparison for $A^A$ contains only functions of the form $(x,v)\mapsto(f(x),a_xv+b_x)$, and therefore misses $(x,v)\mapsto(x,v^2)$. Thus the composite $Q\circ\overrightarrow{\mathsf{CHAD}}$ preserves the first-order product-and-variant structure but is not a cartesian closed interpretation into $\catSet$. Forward CHAD and dual numbers share their first-order derivative behaviour while organising higher-order values differently.

\section{Expressiveness}

Automatic differentiation raises several technical challenges, including the \emph{if-problem} \cite{DBLP:journals/pacmpl/MazzaP21,joss1976algorithmisches}. At a deeper level, however, the central difficulties stem from questions of \emph{expressiveness}. To support realistic programming languages, differentiation must extend beyond its classical setting so that it applies coherently to richer type systems and effects, remains consistent with the chain rule, and agrees with the classical notion on data types. Categorical semantics provides a natural framework in which to address these requirements.

The gap becomes apparent when classical differentiation is contrasted with the requirements of modern programming languages. Classical differentiation is formulated within concrete mathematical categories that lack cartesian closure and do not accommodate non-termination. Functional languages, by contrast, support currying (cartesian closedness), general recursion and iteration (and therefore non-termination), as well as complex types such as recursive types.

Consider, for example, the classical category of manifolds and differentiable maps, $\catM$. This category already fails to capture basic programming constructs. It does not support variant types: if one attempts to define
\[
\texttt{data K = Either } \mathbb{R}^n\; \mathbb{R}^m,
\]
for $n\ne m$, the intended disjoint union $\mathbb{R}^n + \mathbb{R}^m$ is not an object of $\catM$ under our fixed-dimension convention. Categorically, one would like $K$ to satisfy
\[
\catM(K,-) \;\cong\; \catM(\mathbb{R}^n,-)\times \catM(\mathbb{R}^m,-),
\]
but such coproducts do not exist in $\catM$: the two summands would have to be open submanifolds of $K$, contrary to the fixed-dimension convention. Even these simple variant types therefore require a broader setting.

A natural extension is to allow formal families of manifolds, whose components may have different dimensions, by taking the free coproduct completion
\[
\mathsf{VMan} \;:=\; \mathsf{Fam}(\catM).
\]
If $\catC$ has small coproducts, the universal property of the free coproduct completion extends the (co-)differentiation functor \eqref{eq:derivative-functor-question} uniquely up to canonical isomorphism to a coproduct-preserving functor. If the original functor preserves finite products and finite products in $\catC$ distribute over small coproducts, this extension is bicartesian:
\[
\mathsf{D} : \mathsf{VMan}\longrightarrow \catC ,
\]
which we continue to denote by $\mathsf{D}$. In the case of CHAD, this 
amounts to
\[
cha\mathfrak{D} : \mathsf{VMan} \;\longrightarrow\; \mathsf{Fam}(\Vect^{\op}),
\]
which is a bicartesian functor.
This extension better aligns categorical semantics with programming constructs. 

Nonetheless, $\mathsf{VMan}$ still falls short of the expressiveness required by realistic programming languages. For instance, it lacks a general account of partiality and unbounded iteration, as well as function types and general coinductive types. We briefly illustrate these shortcomings in the subsequent discussion.

\paragraph{Higher-order functions.}
In functional programming, functions are first-class values: they can be passed as arguments, returned as results, and composed into new functions. For instance,
\[
\mathsf{applyTwice} : (A \to A) \to (A \to A), \qquad \mathsf{applyTwice}(f)(x) = f(f(x)),
\]
takes a function $f$ and produces a new one. Categorically, this requires \emph{exponential objects}, representing maps from $A$ to $B$ by an object $B^A$ of the category. Such higher-order structure is central to programming languages but absent from both $\catM$ and $\mathsf{VMan}$ by Lemma \ref{rem:VMAN}.

\begin{lemma}\label{rem:VMAN}
	Neither the category $\catM$ nor the category $\mathsf{VMan}$ is cartesian closed.	
\end{lemma} 	

\begin{proof}
Suppose that $E=\mathbb{R}^{\mathbb{R}}$ exists in either category, interpreting $\mathbb{R}$ as a singleton family in $\mathsf{VMan}$. Let $z:1\to E$ represent the zero function. For every $n\geq1$, polynomial interpolation gives a smooth map
\[
p_n:\mathbb{R}^n\times\mathbb{R}\longrightarrow\mathbb{R},\qquad
p_n(x,t)=\sum_{k=1}^n x_k\prod_{\substack{1\leq j\leq n\\j\ne k}}\frac{t-j}{k-j}.
\]
Its transpose $s_n:\mathbb{R}^n\to E$ is a section of the evaluation map $r_n:E\to\mathbb{R}^n$ at $1,\ldots,n$. Moreover, $s_n(0)=z$ for every $n$. In $\mathsf{VMan}$, a morphism from a singleton family selects one component of its codomain, so every $s_n$ selects the component containing $z$. Differentiating $r_n s_n=1$ at zero therefore shows that $\mathbb{R}^n$ is a retract of the same finite-dimensional tangent space $T_zE$. This space would have dimension at least $n$ for every $n$, a contradiction.
\end{proof}

\paragraph{Iteration.}
Real programs also rely on (general) recursion or iteration. Intuitively, iteration repeatedly applies a step until termination. Formally, given
\[
f : A \to B + A ,
\]
where $f$ either produces an output in $B$ or returns a new state in $A$, the \emph{iteration operator} yields a partial function
\[
\mathsf{it}(f) : A \rightharpoonup B ,
\]
defined by
\[
\mathsf{it}(f)(x) = b \quad \text{if some } n \geq 1 \text{ satisfies } \bar f^{\,n}(\iota_2 x)=\iota_1 b,
\]
where $\bar f=[\iota_1,f]:B+A\to B+A$, and the result is undefined otherwise. This formalises looping until a result is returned or divergence occurs (see Chapters~2 and~4).

For example, the factorial function can be expressed via iteration. We take $A=\mathbb{N}\times\mathbb{R}$, with the natural number stored as a discrete counter, and let $B=\mathbb{R}$. The state $(n,a)$ records the remaining counter and the accumulated product. These are objects of $\mathsf{VMan}$ by representing $A$ as a countable family of copies of $\mathbb{R}$. We define
\[
f(n,a) \;:=\;
\begin{cases}
	\iota_1(a) & \text{if } n = 0, \\
	\iota_2(n-1,\, a \cdot n) & \text{if } n > 0 ,
\end{cases}
\]
where $\iota_1,\iota_2$ are coproduct injections. Then $g:\mathbb{N}\to\mathbb{R}$ is defined by
\[
g(n) = \mathsf{it}(f)(n,1)=n!.
\]
\paragraph{General recursion.}
General recursion can be expressed by a \emph{fixpoint operator}
\[
\mathsf{fix}_A : (A \to A) \to A, 
\qquad 
\mathsf{fix}_A(f) = f\big(\mathsf{fix}_A(f)\big).
\]
Intuitively, $\mathsf{fix}_A(f)$ denotes an element of $A$ that is a fixed point of $f$. Such an operator is not definable in $\mathsf{VMan}$, since function types are absent, nor in the simply typed $\lambda$-calculus, which requires additional structure. A standard semantics uses a domain-theoretic setting with least fixed points, thereby accounting for non-termination. A familiar example is the factorial function:
\[
\mathsf{fact}(n) \;=\; 
\begin{cases}
	1 & \text{if } n = 0, \\
	n \cdot \mathsf{fact}(n-1) & \text{if } n > 0,
\end{cases}
\]
which unfolds through recursive calls until the base case is reached. The availability of $\mathsf{fix}$ allows such definitions to be expressed uniformly as $\mathsf{fact} = \mathsf{fix}(F)$, where $F$ maps a candidate $g$ to the function
\[
n \;\mapsto\; \mathbf{if}\; n=0\;\mathbf{then}\;1\;\mathbf{else}\; n \cdot g(n-1).
\]

\paragraph{Recursive types.}
Programming languages also depend on types defined recursively. The most basic examples are \emph{inductive types}, such as natural numbers, lists, and trees, which arise as initial algebras of polynomial functors. For instance, lists over a set $A$ are described by the type equation
\[
\mathsf{List}(A) \;\cong\; 1 \,+\, A \times \mathsf{List}(A),
\]
where $1$ represents the empty list and $A \times \mathsf{List}(A)$ represents cons-cells. Dually, \emph{coinductive types} such as streams are defined as final coalgebras, for example
\[
\mathsf{Stream}(A) \;\cong\; A \times \mathsf{Stream}(A),
\]
capturing potentially infinite sequences. 

In realistic programming, recursion in functions and recursive types often appear together, giving rise to rich structures essential for practical computation. Categorically, this requires objects $K$ defined as solutions to recursive domain equations, so that morphisms $K \to M$ can themselves be specified recursively. Such objects often do not exist in $\mathsf{VMan}$, which therefore lacks the expressive power needed to model these fundamental constructions.

\paragraph{Conclusion.}
From a programming perspective, $\mathsf{VMan}$ lacks several standard constructs and thus cannot serve as a faithful model of real languages. To capture these features adequately, one must move beyond such concrete categories.

\section{Freely generated categorical structures}
One possible approach is to embed $\catM$ into a more expressive category by a completion or another free construction. Such models can be useful, but our immediate objective is a program transformation on an expressive language. Its correctness should recover classical differentiation at the types where that comparison is available, including for programs that use richer constructs internally. This calls for a treatment of syntax and its interpretations, rather than a commitment to a single concrete extension of $\mathsf{VMan}$.

We therefore work at the level of \emph{programming languages modulo $\beta\eta$}, viewing a language as a \emph{freely generated categorical structure} on its primitive types and operations. The universal property specifies the AD macro, while suitable semantic models and logical relations establish its agreement with differentiation.

This stance keeps us focused on the real task: devising transformations from universal properties and proving correctness by structural methods (e.g.\ logical relations), rather than through case analyses tied to a particular category.

\paragraph{Example: simply typed $\lambda$-calculus.}
Given a set $X$ of primitive types and a set $\mathbf{Op}$ of primitive operations (e.g.\ $\mathbf{op}\!: A\times B\to C$), we obtain a typed operation signature
\[
TX \longleftarrow \mathbf{Op} \longrightarrow X
\]
in $\catSet$, where $TX=\coprod_{n\geq0}X^n$ is the set of finite lists of primitive types, recording the ordered input types of an operation. More general operation signatures allow arbitrary type expressions. The free cartesian closed category on the signature presents the simply typed $\lambda$-calculus via the Curry--Howard--Lambek correspondence: tuples and function types arise as products and exponentials.

Grounded in two-dimensional universal algebra, the same construction extends to richer languages. Reasoning about programs modulo $\beta\eta$ becomes reasoning in their freely generated structure. Two features are essential: (i) the structure reflects the language’s expressive power through familiar categorical semantics, and (ii) the universal property ensures that \emph{structure-preserving functors}, the semantics of homomorphic program transformations, are uniquely determined by their action on primitives.

\paragraph{Working hypothesis.}
Freely generated structures provide a principled bridge between abstract (categorical) semantics and implementation concerns. This thesis substantiates that claim: treating programming languages modulo $\beta\eta$ as free categorical structures produces principled, practically relevant, and verifiable program transformations while enriching categorical semantics.

\paragraph{AD via freeness.}
In particular, modelling AD as a homomorphic transformation on the programming language modulo $\beta\eta$  yields principled, modular, and correct AD. Correctness follows from universal properties and categorical logical relations, and optimisation opportunities emerge more clearly. In this way, AD serves as a sustained test of the methodology, aligning categorical semantics with real computational demands.

Finally, instead of enlarging concrete models piecemeal, we formulate differentiation \emph{within} the appropriate free structures. This ensures compatibility with standard models whenever applicable, while supporting the expressiveness required by modern programming languages.

\section{AD from universal properties}

The guiding principle of this thesis is that differentiation should be \emph{structure-preserving}. This perspective is consistent with classical mathematics, where tangent bundles and dual numbers exhibit functorial behaviour, and extends naturally to expressive programming languages when syntax is treated as free categorical structure.

Formally, let a language be modelled by a categorical structure $\mathcal{C}$ (e.g.\ a cartesian closed, indexed, or Freyd category). The programming language modulo $\beta\eta$ is then the freely generated $\mathcal{C}$-structure on a set of primitives/signatures. Any uniform transformation that respects this structure corresponds to a \emph{structure-preserving functor}/morphism, uniquely determined by its action on primitives. Instead of specifying AD via \textit{ad hoc} syntactic rules, we characterise it as the unique homomorphic extension induced by suitable semantic data.

This framework encompasses both forward- and reverse-mode AD.  
For forward mode, we recover dual numbers by defining a structure-preserving functor from the syntactic category of a simply typed higher-order source language to that of an extended target language, sending each primitive operation to code for its value and derivative action. The functor extends homomorphically to all programs, and correctness is established via a categorical logical-relations argument that accommodates recursive types and partial primitives.

CHAD provides a second approach to both forward- and reverse-mode AD. We develop \textbf{CHAD} (\emph{Combinatory Homomorphic AD}) for a higher-order total language with rich data types (products, sums, inductive and coinductive types). The semantics is expressed via indexed categories and their Grothendieck constructions, which accumulate derivative information through data projections. To handle sums systematically we introduce \emph{extensive indexed categories}, a fibred analogue of extensive categories. The reverse transformation is then given by a homomorphic functor into a Grothendieck construction, with correctness following from universal properties and refined logical relations.

We further extend the approach to first-order languages with partial operations and general while-loops, yielding \emph{Iterative CHAD}. Here, \emph{iteration-extensive indexed categories} relate substitution along a loop to parameterised initial algebras of substitution along its body. Their folds supply iteration in the category of containers; \emph{iterative folders} isolate the operations required by the target syntax. The transformation remains structure-preserving, and its correctness is proved compositionally.

\medskip
\noindent\textbf{The methodological connection.}
Throughout, language features determine the categorical structure to be preserved. Universal properties specify the transformations, and semantic comparisons establish their correctness.
\section{Free extensivity and distributivity}

Throughout this research we have emphasised that freely generated categorical structures are central to categorical semantics: a programming language modulo $\beta\eta$ can be seen as a \emph{freely generated categorical structure}. Concrete models then appear as non-free algebras, with syntax represented by the free counterpart. This perspective has long been fruitful in program semantics, and it has proven especially powerful in the development of AD.

Pursuing AD from this vantage point raised questions that extend beyond the applied setting into foundational aspects of category theory. In formulating AD as a homomorphic transformation, we encountered structural requirements on free completions under products and coproducts. Within two-dimensional universal algebra, this prompted a systematic study of the free completions under small products and small coproducts (Chapter~5), followed by analogous constructions for classes of limits (Chapter~6). A canonical pseudodistributive law between the product and coproduct completions yields a pseudomonad whose pseudoalgebras are the \emph{doubly-infinitary distributive categories}.

Concretely, a category $\catC$ with small products and small coproducts is doubly-infinitary distributive when the canonical functor
\[
\Sigma \colon \Fam{\catC} \to \catC
\]
preserves small products, where $\Fam{\catC}$ denotes the free coproduct completion of $\catC$. This condition is strictly stronger than ordinary infinitary distributivity. For instance, $\catTop$ is infinitary distributive but not doubly-infinitary. Moreover, although coproducts together with cartesian closedness imply infinitary distributivity, this implication fails at the doubly-infinitary level: the category of pseudotopological spaces $\mathsf{PsTop}$ is complete, cocomplete, and cartesian closed (indeed a quasitopos), yet the inclusion $\mathsf{Top}\to \mathsf{PsTop}$ preserves both coproducts and products, showing that $\mathsf{PsTop}$ is not doubly-infinitary distributive. Chapter~5 develops these comparisons in detail.

Our main focus, however, lies on the freely generated structures themselves. We prove that free doubly-infinitary distributive categories are cartesian closed: an unexpected result in which function types arise purely from freeness, becoming an \emph{admissible} property (alongside, for instance, extensivity). Although initially motivated by the applied challenge of differentiable programming and optimisation issues in CHAD, this theorem is of independent interest: it clarifies the interplay between distributivity and two-dimensional monad theory, underpins later developments on free extensive categories (Chapter~6), and points toward new connections between freely generated structures, function spaces, and categorical accounts of defunctionalisation (left to future work).

Thus, a concrete problem in automatic differentiation gave rise to new results in pure category theory, which in turn enrich the research program on homomorphic program transformations. The broader lesson is that categorical investigations inspired by AD are not confined to differentiation alone but illuminate the semantics of programming languages and can be directly applied to a wide range of principled, structure-preserving transformations.
\section{Grothendieck constructions}

A second strand of foundational contributions concerns \emph{lifting structure through Grothendieck constructions}.  
The guiding question is how much of the structure of a base category can be transported along a fibration to its total category.  
Within CHAD this issue became central when analysing the behaviour of exponentials, both cartesian and monoidal, as well as coproducts, initial algebras, terminal coalgebras, and even iteration, of the Grothendieck constructions of indexed categories.  
A significant portion of the theory needed for CHAD reduces to clarifying these aspects of fibred category theory, which has to do with completion aspects of the $2$-category of indexed categories.

These lifting questions are equally relevant to logical relations methods.  
The categorical logical relations technique can be viewed as a semantics obtained by transporting the structure of a base model into a category of predicates (see, for instance, \cite{nunes_vakar_mfo_4046}).  
Constructing such relational models depends on principled ways of lifting limits, colimits, initial algebras, monads, and exponentials along functors that are often fibrations.

To treat the closure constructions required here in a common framework, we introduce \emph{$\Sigma$-(co)tractable} monoidal structures. Indexed cotractability, together with suitable $\Pi$-types and base-category hypotheses, yields a Dialectica-inspired construction of internal homs in the Grothendieck construction.  
This provides a uniform route to total categories supporting function spaces, built from indexed families.

The resulting framework unifies many constructions of monoidal closed and cartesian closed total categories from indexed categories.
As in other parts of the thesis, concrete needs from AD and CHAD prompted the questions, and the answers feed back as independent contributions to fibred category theory.  
Practical concerns about reverse-mode differentiation thus led to a systematic account of when and how structure lifts through Grothendieck constructions, with value beyond the original application.

\section{Extensivity and related contributions}

A recurring theme in this work is \emph{extensivity}.  
For many programming languages considered modulo $\beta\eta$, the associated free categorical structures are extensive, a fact that repeatedly guided our semantic designs.

To capture this behaviour in a fibred setting and to handle coproducts, we introduced \emph{extensive indexed categories} in \cite{nunes_vakar_2021_chad} (Chapter~3).  
This notion ensures that coproducts lift through Grothendieck constructions, which in turn allows variant types to be modelled cleanly in dependently typed semantics.  
The idea is developed further in \cite{nunes_iterative_chad_2025} (Chapter~4), where we formulate an analogous principle for iteration.  
There we reconcile general while-loops with dependent types by requiring that substitution along an iterated computation be isomorphic to the parameterised initial-algebra carrier functor induced by substitution along its body.

In parallel, the notion of \emph{$\Sigma$-tractable} coproduct structures from Chapter~7 can be read as a generalisation of the decomposition of morphisms supplied by extensivity.  
It controls the behaviour of morphisms with coproduct codomains and clarifies the interaction between coproducts and exponentials inside Grothendieck constructions.

Taken together, these results present extensivity as a unifying principle that underlies coproducts, iteration, lifting, and monoidal closure in our semantics.  
The trajectory is characteristic of the thesis as a whole: requirements that arise from CHAD, such as modelling sums and loops in expressive typed settings, lead to general theorems about extensive and fibred categories.  
Applied questions in differentiable programming thus stimulate advances in categorical foundations, and those advances in turn strengthen the semantic tools available for principled program transformations.

\section{Contributions and Thesis Outline}

The development of \emph{structure-preserving differentiation} yields new results in categorical semantics that stand on their own. The dissertation interleaves practical transformations with these foundational advances.

\medskip
\noindent\textbf{The full thesis is structured as follows:}
\begin{itemize}
	\item \textbf{Chapter 2.} Dual-Numbers AD for ML-family Languages. This chapter corresponds to \cite{nunes_vakar_2024_ml}. We study forward- and reverse-mode dual-numbers AD for an ML-family language (simply typed $\lambda$-calculus with recursive types).  
	It gives an account of dual-numbers AD in both modes for a language with general recursive types, together with a unified semantic correctness proof.

	We present the dual-numbers transformation as a structure-preserving morphism between the freely generated categorical structures associated with the source and target languages.  
	A streamlined logical-relations argument establishes correctness: differentiating syntax yields the same derivative as the denotational semantics of the program.  
	This provides a template for constructing and verifying program transformations in languages with recursive types.  
	Since Pitts’ logical-relations framework for recursive types~\cite{DBLP:journals/iandc/Pitts96,deamorim2022pittsrelationalpropertiesdomains} remains the classical reference in this area, a detailed comparison between our approach and his is left to future work.  
	
	The treatment includes iteration, term recursion, and recursive types. The transformation is specified from the syntactic structure, while the concrete semantics interprets recursion by least fixed points.  
	
	Finally, we derive correct reverse-mode dual-numbers AD from the same transformation using a suitable vector-space interpretation and input wrappers, again for languages with recursive types.  
	The companion implementation work of Smeding and V\'ak\'ar studies efficient dual-numbers reverse AD and the preservation of parallelism~\cite{DBLP:conf/popl/SmedingV23,smeding2022parallel}. These developments complement the semantic correctness results presented here.

	\item \textbf{Chapter 3.} CHAD for expressive total languages. 
	This chapter corresponds to \cite{nunes_vakar_2021_chad}.  
	
	It presents the first categorical account of reverse-mode AD for languages that include not only products and higher-order functions but also coproducts, inductive types, and coinductive types. The same framework also covers forward-mode AD.  
	The framework is fully compositional and avoids both explicit mutation and graph tracing.  
	Its central theoretical contribution is the introduction of \emph{extensive indexed categories}, which capture disjoint sums in a fibred setting. We also introduce related notions that ensure initial algebras and terminal coalgebras lift through the fibration, as well as results on the creation of terminal coalgebras by monadic functors and initial algebras by comonadic functors, for compatible endofunctors.
	
	Within this setting, forward- and reverse-mode AD are formulated as homomorphic transformations, with correctness established via a refined categorical logical-relations proof that exploits monadicity and comonadicity techniques in analysing lifted initial algebras and terminal coalgebras.  
	
	The efficient implementation of CHAD developed by Smeding and V\'ak\'ar~\cite{DBLP:conf/popl/SmedingV24} complements this semantic account. The relationship between the dependent presentation and a simply typed target is examined further in the recent work discussed in Chapter~8.

	\item \textbf{Chapter 4.} Unraveling the iterative CHAD. 
	This chapter corresponds to \cite{nunes_iterative_chad_2025}.  
	
	Iterative CHAD extends CHAD to cover first-order programs with general iteration (while-loops), non-termination, and partial (potentially undefined) operations, while preserving its structure-preserving semantics.  
	
	We introduce \emph{iteration-extensive indexed categories} to model these features in a dependently typed setting.  
	We relate substitution along a loop to parameterised initial algebras of substitution along its body. The resulting folds induce fibred iteration on the op-Grothendieck construction. Iterative folders retain the operations needed by the target syntax, while the enriched models identify the induced iteration with least-fixed-point iteration.  
	Through these categorical constructions and a syntactic categorical model with a universal property, we define a unique homomorphic extension of CHAD in the presence of iteration and non-termination.  
	We then prove a correctness theorem: the differentiated program has the same domain of definition as the original program and computes its reverse-mode derivative wherever it is defined.  
	
	Beyond its theoretical contributions, this work advances the general study of dependently typed languages with iteration.  
	It shows new techniques for reconciling iteration and dependent types, with further applications expected in program logics.  
	
	\item \textbf{Chapter 5.} Free Doubly-Infinitary Distributive Categories.  
	This chapter corresponds to \cite{nunes_vakar_2024_doubly_inf_distrib}.  
	
	We construct the pseudomonad \textsf{Dist}, arising from the canonical pseudo-distributive law between the free completion pseudomonads under small products and small coproducts.  
	Its pseudoalgebras are the \emph{doubly-infinitary distributive categories}, a class of categories in which the canonical functor $\Sigma \colon \Fam{\catC} \to \catC$ preserves small products.  
	This property is strictly stronger than ordinary infinitary distributivity and had not been systematically studied in the literature.  
	
	Our main result is that the free doubly-infinitary distributive completion of any category $C$ is cartesian closed, and therefore supports function types.  
	This is surprising: while free (infinitary) distributive categories generally fail to be cartesian closed, the stronger distributivity of the doubly-infinitary setting makes cartesian closure an \emph{admissible} property of the free construction.  
	We also give explicit formulas for exponentials in $\mathrm{Dist}(C)$, whose combinatorial flavour resonates with Dialectica-style constructions.  
	
	The work originated from attempts to understand programming languages up to $\beta\eta$-equivalence and led to unexpected foundational insights.  
	It reveals new connections between distributivity, freeness, and the categorical representation of higher-order functions, and points toward a categorical formulation of defunctionalisation~\cite{reynolds1972definitional}. 
	Although motivated by differentiable programming, the result stands independently as a contribution to two-dimensional monad theory and to the study of free categorical structures.

	\item \textbf{Chapter 6.} Free extensivity via distributivity.  
	This chapter corresponds to \cite{nunes_extensivity_2025}, and builds directly on \cite{nunes_vakar_2024_doubly_inf_distrib}.  
	We show that a canonical pseudodistributive law between the free completion under coproducts and under suitable limits gives rise to a pseudomonad whose pseudoalgebras correspond to natural classes of extensive categories.  
	
	This result deepens the relationship between distributivity and extensivity by identifying extensivity as a property that emerges from universal constructions.  
	In particular, we prove that free doubly-infinitary lextensive categories are cartesian closed, extending the perspective of Chapter~5 and revealing new structural compatibilities between sum types and function types.  
	
	As before, the motivation stems from programming-language semantics: understanding languages modulo $\beta\eta$ and clarifying the role of extensivity as an admissible property led directly to new foundational results in category theory.

\item \textbf{Chapter 7.} Grothendieck constructions. This chapter corresponds to \cite{nunes_monoidal_grothendieck_2024}.  
The guiding question is: \emph{how much structure of a category can be lifted through a fibration?}  
In particular, we study the preservation of exponentials, both in cartesian and in monoidal contexts.  

This investigation originated from our work on CHAD, which required a systematic understanding of how properties such as closure and extensivity are transported through Grothendieck constructions. Building on the classical literature on fibrations, we develop conditions that apply uniformly to the constructions needed here.  

We begin by analysing fibred limits, colimits, and monoidal (closed) structures of Grothendieck constructions. In doing so, we extend the notion of \emph{extensive indexed category} introduced in Chapter~3, giving rise to the concept of \emph{left Kan extensivity}. This new notion provides a uniform mechanism for computing colimits in Grothendieck constructions and refines the understanding of how coproducts behave under fibred colimit formation.  

Building on this, we introduce \emph{$\Sigma$-(co)tractable} monoidal structures. Indexed cotractability, together with suitable $\Pi$-types and base-category hypotheses, gives sufficient conditions for monoidal closure through a Dialectica-inspired formula. Under this framework, $\Sigma$-tractable coproducts unify and extend cocartesian coclosed structures, biproducts, and extensive coproducts.  

Once again, concrete challenges from AD and CHAD led to advances of independent value in pure category theory, expanding the theory of extensive indexed categories through the introduction of left Kan extensivity and a clearer understanding of exponentials in Grothendieck constructions.

	\item \textbf{Chapter 8.} Conclusion and future work. We reflect on how the dialogue between automatic differentiation and categorical semantics led to the results above, summarise contributions and limitations, and outline directions for program transformations beyond differentiation.
\end{itemize}

\noindent
Universal properties provide a common framework for the transformations studied in this thesis, while categorical logical relations and semantic comparisons establish their correctness. The resulting development contributes both to automatic differentiation for expressive languages and to the study of free and fibred categorical structures. It also suggests a broader programme of structure-preserving program transformations, with AD as its central case study.

\setcounter{chapter}{7}
\chapter{Conclusion and Future Work}

This thesis has examined the interaction between categorical semantics and program transformation through automatic differentiation (AD). The practical requirements of AD led us to account for linear transformations, higher-order functions, recursive types, sums, iteration, and partiality within a compositional framework. Treating programming languages modulo their equations as \emph{freely generated categorical structures} allowed us to define AD by its action on primitives and to organise correctness proofs through universal properties and categorical logical relations. The same investigation raised questions about free completions and Grothendieck constructions, leading to results of independent categorical interest.

\section*{From Foundations to Homomorphic AD}

On the programming-language side, the thesis developed a framework for \emph{homomorphic} (\textit{i.e.} structure-preserving) program transformations, with AD as its main application. Three case studies illustrate the methodology:

\begin{itemize}[itemsep=4pt,parsep=0pt,topsep=6pt]
	\item \textbf{Dual-numbers AD (forward and reverse, Ch.~2).}  
	Dual-numbers AD was established as a structure-preserving morphism between the syntactic models of an ML-family source language and an extended target language, both with recursive types. We established correctness by a reusable categorical logical-relations argument accommodating recursive types, partial primitives and differentiability/local properties~\cite{nunes_vakar_2024_ml}.

	\item \textbf{Forward- and reverse-mode AD: CHAD (Ch.~3).}  
	We extended \emph{Combinatory Homomorphic AD} (CHAD) to expressive total languages with tuples, variant types, inductive and coinductive types, and higher-order functions. The dependent target semantics led us to introduce \emph{extensive indexed categories} to model sums in the Grothendieck construction. Correctness followed from universal properties and categorical logical-relations proofs encompassing initial algebras and terminal coalgebras, using the monadicity and comonadicity results developed in the chapter~\cite{nunes_vakar_2021_chad}.
	
	\item \textbf{Iteration and partiality: Iterative CHAD (Ch.~4).}  
	Extending the first-order fragment of CHAD to general iteration and non-termination led to \emph{iteration-extensive indexed categories}. These relate substitution along a loop to parameterised initial algebras of substitution along its body. Their folds yield fibred iteration on the op-Grothendieck construction, and iterative folders express the corresponding operations in the target language. This distinction between semantic universal properties and syntactic operations supports a homomorphic extension of CHAD with a proof of domain preservation and derivative correctness~\cite{nunes_iterative_chad_2025}.
\end{itemize}

\section*{From AD to Foundations}

Each extension of AD demanded new categorical machinery, leading to stand-alone foundational results:
\begin{itemize}[itemsep=4pt,parsep=0pt,topsep=6pt]
	\item \textbf{Extensivity and distributivity (Ch.~5 and Ch.~6).}  
	By studying a pseudodistributive law between the free completions under products and coproducts, we arrived at the notion of a \emph{doubly-infinitary distributive category}. We proved that the free such categories are cartesian closed.  
	
	Building on this, by analysing pseudodistributive laws between the free completions under limits and under coproducts, we obtained notions of \emph{free lextensive completions}. The free doubly-infinitary lextensive completions are cartesian closed; in the free infinitary lextensive case, objects with finitely many connected components are exponentiable.  
	
	\item \textbf{Grothendieck constructions (Ch.~3, Ch.~4, and Ch.~7).}  
	In Chapters~3 and~4, we introduced several notions of extensive indexed categories, providing a clear understanding of coproducts, iteration, and even parameterised initial algebras (and terminal coalgebras) within their respective Grothendieck constructions.  
	
	In Chapter~7, we studied limits, colimits, and monoidal closure. Indexed \emph{$\Sigma$-cotractability}, together with suitable $\Pi$-types and base-category hypotheses, gives sufficient conditions for monoidal closure of the total category. We also developed \textit{left Kan extensivity} to describe the lifting of colimits to Grothendieck constructions.  
\end{itemize}

These results were inspired by the challenges encountered while devising principled AD macros. They also contribute independently to the understanding of extensive categories, fibred category theory, and two-dimensional universal algebra.

\section*{Limitations}

Alongside providing a solution to the problem of reverse-mode AD for the expressive programming languages studied here, this work has demonstrated that categorical semantics can be employed to design practically relevant program transformations. Our emphasis on correctness and clear semantics was not merely a matter of theoretical elegance: it showed that abstract categorical reasoning can yield concrete AD macros whose derivative correctness is established by semantic proofs. At the same time, this foundation lays the groundwork for systematic reasoning about implementation concerns, such as optimisation and efficiency, by providing a clean and principled semantic framework.

Building on existing work on efficient implementations, the natural next frontiers lie in addressing concrete implementation issues. These include remaining efficiency challenges, the role of defunctionalisation, and the question of how CHAD might further exploit parallelism. While such matters belong partly to the practical engineering of AD systems, categorical structures can provide a sharper and more principled perspective. In particular, there are opportunities to study operational aspects (e.g. rewriting systems), to analyse defunctionalisation through its categorical semantics, and to explore structured approaches to parallelism in CHAD.

Ultimately, the open question is how far this methodology can be taken. This thesis has shown that categorical semantics, when held accountable to practice, is not only explanatory but generative: it can deliver transformations that are provably correct and practically relevant. The challenge now is to push this dialogue further: to see whether the same categorical reasoning that has secured correctness can also chart a systematic roadmap to efficiency, optimisation, and implementation. In this way, theory and practice may continue to advance together, each strengthening and extending the other.

\section*{Recent Developments}

Work completed since the original submission has already advanced two directions suggested by this thesis: extending CHAD to more expressive languages, and understanding more precisely the relationship between its semantic account and its implementation. The two papers discussed below are subsequent work and are not included among the contributions of this thesis.

In the first direction, our joint work with Diogo Simm and Matthijs V\'ak\'ar, \emph{Backpropagation for Effectful Languages I: Finite Probability and Discrete Output Algebraic Effects}~\cite{simm_nunes_vakar_2026_effectful}, extends the CHAD approach to finite-support probabilistic computations with differentiable weights. Reverse propagation must account for the weights as well as the deterministic computations. The construction treats monadic operations and handlers compositionally and proves correctness for handled programs with real-valued outputs. This provides a concrete step towards richer probabilistic languages and other algebraic effects; continuous probability and more general combinations of effects still require further work.

In the second direction, joint work with Diogo Simm and Matthijs V\'ak\'ar~\cite{nunes_simm_vakar_2026_variants} gives a simply typed presentation of dependent CHAD for products, variants and higher-order functions. For variants, primal-indexed idempotents represent the dependent cotangent fibres inside ordinary types. Hoofman and Moerdijk's account of semifunctors and Cauchy completion~\cite{hoofmanMoerdijk1995semifunctors} explains why the ambient transformation preserves composition but sends identities to idempotents, while its interpretation in the completion is an ordinary functor. This relates implementable program representations more directly to the semantic principles developed in this thesis, reinforcing the prospect of advancing theory and practice together.

\section*{Future Work}

The contributions of this thesis point naturally towards two complementary strands of further research:  
(A) advancing program transformations and differentiable programming, and  
(B) developing a categorical toolkit for programming language semantics and program transformations more broadly.  

\Needspace{6\baselineskip}
\paragraph{(A) Program transformations and differentiable programming.}

\begin{itemize}[itemsep=4pt,parsep=0pt,topsep=6pt]
	\item \textbf{Extending CHAD.} Extend the CHAD account to ML-family languages, developing the treatment of recursion and recursive types in an indexed target or an appropriate simply typed presentation.  
	
	\item \textbf{Higher-order differentiation.} Extend the homomorphic approach to higher-order reverse-mode differentiation, moving beyond first derivatives to capture higher-order derivatives in a principled way.  
	
	\item \textbf{Simply typed CHAD.} Extend the idempotent-completion account of variants discussed above to richer language fragments, and study how its semantic comparison supports optimisations such as defunctionalisation, closure conversion, and sharing.  
	
	\item \textbf{AD with effects.} Build on the finite-probability and discrete-output account described above, extending CHAD to richer probabilistic languages and combinations of algebraic effects and handlers while retaining compositional semantics and correctness proofs.  
\end{itemize}

\Needspace{6\baselineskip}
\paragraph{(B) A categorical toolkit for PL semantics and program transformations.}

\begin{itemize}[itemsep=4pt,parsep=0pt,topsep=6pt]
	\item \textbf{Presentations via 2-dimensional universal algebra.} Systematise the theory of presentations of categorical structures for programming-language semantics, with particular attention to coherence and to comparisons between alternative categorical models of the same constructs.  
	
	\item \textbf{Beyond AD.} The thesis suggests an underlying principle for transformations whose auxiliary information depends on the program value. Ongoing work investigates how such value-indexed semantic data can give rise to uniform, structure-preserving transformations beyond differentiation. Their further applications and implementation remain a direction for future research.  
\end{itemize}

\section*{Closing}

This thesis has shown how categorical semantics can yield practical methods for constructing correct program transformations. We derived compositional AD transformations for expressive languages, and the CHAD account developed here formed part of the foundation for Smeding and V\'ak\'ar's efficient implementation~\cite{DBLP:conf/popl/SmedingV24}. Conversely, the requirements of AD led to independent mathematical results on free completions, extensivity and Grothendieck constructions.

Subsequent work on simply typed CHAD further clarifies the relation between these semantic foundations and their implementation. These developments support a promising programme of research in which categorical models, correctness proofs and efficient implementations are developed together. The thesis thus substantiates the principle with which it began: sound semantic foundations can guide the principled construction of reliable and practically useful program transformations, with explicit correctness guarantees, while the demands of computation can lead to new mathematics.

\backmatter
\cleardoublepage
\phantomsection

\chapter{Samenvatting}
De categorietheorie biedt een krachtig fundament voor de semantiek van programmeertalen en maakt het mogelijk om typesystemen, programma’s en programequivalenties wiskundig nauwkeurig te modelleren. 
Klassieke werken zoals J.~Lambek en P.~Scott’s \textit{Introduction to Higher-Order Categorical Logic}, R.~Crole’s \textit{Categories for Types}, en B.~Jacobs’ \textit{Categorical Logic and Type Theory} tonen aan hoe categorische structuren een breed scala aan type-theorieën en programmeertalen kunnen modelleren. 
Voortbouwend op deze fundamenten introduceerde E.~Moggi de monaden, waarmee hij een verenigend categorisch kader bood voor het representeren van computationele effecten, zoals toestand, uitzonderingen en niet-determinisme. 
Sindsdien is de invloed van de categorietheorie alleen maar toegenomen, zowel bij het analyseren van bestaande taalfaciliteiten als bij het inspireren van nieuwe. 
Zo identificeerden M.~Fiore en G.~Plotkin in hun werk over axiomatische domeintheorie categorische voorwaarden, zoals algebraïsche compactheid, die voldoende zijn om recursieve domeinvergelijkingen op te lossen en zo een categorische generalisatie van domeinsemantiek te verkrijgen. 
Evenzo introduceerden G.~Plotkin en J.~Power met hun theorie van algebraïsche effecten een modulaire benadering om effecten te behandelen via algebraïsche theorieën, waarmee zij Moggi’s monadische perspectief verfijnden en uitbreidden. 
Meer recent verenigde P.~B.~Levy’s \textit{Call-by-Push-Value} de evaluatiestrategieën \textit{call-by-name} en \textit{call-by-value} binnen één calculus, ondersteund door een elegante categorische semantiek gebaseerd op een adjunctiedecompositie van Moggi’s kader. 
Gezamenlijk tonen deze ontwikkelingen aan dat de categorietheorie een robuust en verenigend fundament blijft bieden voor de semantiek van programmeertalen.\par\medskip
\textbf{Dit proefschrift: } Voortbouwend op deze traditie ontwikkelt dit proefschrift de semantiek van programmeertalen verder door middel van categorische structuren, telkens geleid door een concreet en praktisch probleem: \emph{structuurbehoudende programmotransformaties voor het berekenen van afgeleiden, een techniek die algemeen bekendstaat als automatische differentiatie (AD)}. 
Onze theorie wordt ontwikkeld in nauwe samenhang met implementatievraagstukken, 
zodat onze categorische constructies realistische programmotransformaties opleveren. 
Deze transformaties zijn compositioneel; hun correctheid wordt bewezen met categorische logische relaties. Zij bieden tevens een basis voor verdere analyse van hun complexiteit.\par\medskip
\noindent
Onze centrale benadering beschouwt de syntaxis van een programmeertaal, modulo haar standaard $\beta\eta$-equationele theorie, als een geschikte \emph{vrij gegenereerde} categorische structuur. 
Programmotransformaties die de structuur van programma’s behouden, corresponderen dan met \emph{structuurbehoudende (of homomorfe) functoren} die voortkomen uit de universele eigenschappen van deze vrije constructies. 
Volgens dit principe definiëren we \emph{voorwaartse automatische differentiatie} in de stijl van duale getallen als een structuurbehoudende functor van de syntactische categorie van een recursieve hogere-ordebrontaal naar die van een uitgebreide doeltaal, waarbij elke primitieve operatie wordt afgebeeld op code voor haar waarde en de werking van haar afgeleide. 
We bewijzen de correctheid met behulp van een nieuwe techniek van \emph{logische relaties} voor talen met recursieve types en partiële primitieve operaties.\par\medskip
\noindent
Vervolgens breiden we dit homomorfe perspectief op AD uit naar de meer uitdagende en veelgebruikte \emph{achterwaartse modus} van AD, door middel van \textbf{CHAD} (\textit{Combinatory Homomorphic AD}), die \emph{totale} hogere-ordeprogrammeertalen behandelt met expressieve constructies: producten, sommen, inductieve en coïnductieve types, en functietypes. 
De semantiek van CHAD wordt ontwikkeld met behulp van de totale categorieën (Grothendieck-constructies) die horen bij \emph{geïndexeerde categorieën}, en introduceert het begrip \emph{extensieve geïndexeerde categorieën} om disjuncte sommen te behandelen. 
De correctheid volgt uit de universele eigenschap en een geschikte categorische techniek van logische relaties.\par\medskip
\noindent
We breiden CHAD verder uit naar \emph{partiële} eerste-ordetalen met algemene iteratie. 
De \emph{Iteratieve CHAD} integreert iteratie via \emph{iteratie-extensieve geïndexeerde categorieën}, waardoor we onbegrensde lussen kunnen interpreteren en aantonen dat de transformatie opnieuw uniek structuurbehoudend en correct is.\par\medskip
\noindent
In de laatste hoofdstukken tonen we aan dat veel van het wiskundige instrumentarium dat ontwikkeld werd om CHAD te begrijpen, zich op natuurlijke wijze uitbreidt tot fundamentele resultaten in de vezelcategorietheorie, waarbij de structuur van Grothendieck-constructies van geïndexeerde categorieën wordt verduidelijkt, evenals in de theorie van distributieve en extensieve categorieën.\par\medskip

Het volledige proefschrift is als volgt opgebouwd:
\begin{enumerate}[leftmargin=2em,itemsep=6pt,parsep=0pt]
	\Needspace{5\baselineskip}
	\item \textbf{Inleiding}
	\Needspace{5\baselineskip}
	\item \textbf{Automatic Differentiation for ML-family Languages}\\
	(\href{https://arxiv.org/pdf/2210.07724}{arXiv:2210.07724})\\
	Dit hoofdstuk komt overeen met \textit{F.~Lucatelli Nunes} en \textit{M.~Vákár}, 
	``Automatic Differentiation for ML-family languages: correctness via logical relations,'' 
	\textit{Mathematical Structures in Computer Science} \textbf{34}(8), 747--806 (2024), 
	ook beschikbaar via \href{https://arxiv.org/abs/2210.07724}{arXiv:2210.07724}. 
	We ontwikkelen voorwaartse en achterwaartse automatische differentiatie (AD) via duale getallen voor ML-familietalen met recursieve types en bewijzen de correctheid met behulp van logische relaties. 
	De AD-macro wordt geformuleerd als een structuurbehoudend morfisme tussen de syntactische modellen van de bron- en doeltaal, beide met recursieve types. 
	Voor het correctheidsbewijs introduceren we een gestroomlijnd en herbruikbaar raamwerk van logische relaties voor redeneren over recursieve types en talen met effecten.
	
	\Needspace{5\baselineskip}
	\item \textbf{CHAD for Expressive Total Languages}\\
	(\href{https://arxiv.org/pdf/2110.00446}{arXiv:2110.00446})\\
	Dit hoofdstuk komt overeen met \textit{F.~Lucatelli Nunes} en \textit{M.~Vákár}, 
	``CHAD for expressive total languages,'' 
	\textit{Mathematical Structures in Computer Science} \textbf{33}(4--5), 311--426 (2023), 
	\href{https://doi.org/10.1017/S096012952300018X}{doi:10.1017/S096012952300018X}, 
	ook beschikbaar via \href{https://arxiv.org/abs/2110.00446}{arXiv:2110.00446}. 
	We ontwikkelen CHAD voor totale talen met sommen, producten, inductieve, coïnductieve en functietypes. 
	Om dit te bereiken integreren we inductieve en coïnductieve types in afhankelijke typesemantiek en verduidelijken hun relatie via het nieuwe begrip \emph{extensieve geïndexeerde categorieën}. 
	Daarnaast introduceren we krachtige categorische technieken van logische relaties voor redeneren over programmeertalen, gebaseerd op resultaten over monadiciteit en comonadiciteit.
	
	\Needspace{5\baselineskip}
	\item \textbf{Unraveling the Iterative CHAD}\\
	(\href{https://arxiv.org/pdf/2505.15002}{arXiv:2505.15002})\\
	Dit hoofdstuk komt overeen met \textit{F.~Lucatelli Nunes}, \textit{G.~Plotkin} en \textit{M.~Vákár}, 
	``Unraveling the Iterative CHAD,'' 
	\textit{Preprint}, \href{https://arxiv.org/abs/2505.15002v3}{arXiv:2505.15002v3} (2025), herzien op 13 september 2026. 
	We breiden CHAD uit naar eerste-ordetalen met iteratie en niet-terminatie, waarbij we fixpunten integreren via het begrip \emph{iteratie-extensieve geïndexeerde categorieën}, dat hierin wordt geïntroduceerd. 
	Daartoe ontwikkelen we verschillende nieuwe vormen van \emph{extensieve geïndexeerde categorieën} en, belangrijker nog, het concept van \emph{iteratie-extensieve geïndexeerde categorieën}, die de categorische basis bieden voor bepaalde colimieten en iteratie in de overeenkomstige Grothendieck-constructies. 
	De bijbehorende \emph{iteratieve folders} beschrijven de operaties die nodig zijn om de transformatie in de doeltaal uit te drukken.
	
	\Needspace{5\baselineskip}
	\item \textbf{Free Doubly-Infinitary Distributive Categories}\\
	(\href{https://arxiv.org/pdf/2403.10447}{arXiv:2403.10447})\\
	Dit hoofdstuk komt overeen met \textit{F.~Lucatelli Nunes} en \textit{M.~Vákár}, 
	``Free Doubly-Infinitary Distributive Categories are Cartesian Closed,'' 
	\textit{Applied Categorical Structures} \textbf{34}, artikel~53 (2026)~\cite{nunes_vakar_2024_doubly_inf_distrib}, ook beschikbaar via \href{https://arxiv.org/abs/2403.10447}{arXiv:2403.10447}. 
	Onder andere introduceren we het concept van \emph{dubbel-oneindig distributieve categorieën}, dat wil zeggen categorieën waarin willekeurige kleine producten distribueren over willekeurige kleine coproducten, in de precieze zin die in Hoofdstuk~5 wordt ontwikkeld. 
	We bewijzen verder dat vrij dubbel-oneindig distributieve categorieën cartesiaans gesloten zijn.
	
	\Needspace{5\baselineskip}
	\item \textbf{Free Extensivity via Distributivity}\\
	(\href{https://arxiv.org/pdf/2405.02185}{arXiv:2405.02185})\\
	Dit hoofdstuk komt overeen met \textit{F.~Lucatelli Nunes}, \textit{R.~Prezado}, en \textit{M.~Vákár}, 
	``Free Extensivity via Distributivity,'' 
	\textit{Portugaliae Mathematica} \textbf{82}(1--2), 177--204 (2025), 
	ook beschikbaar via \href{https://arxiv.org/abs/2405.02185}{arXiv:2405.02185}. 
	We bestuderen de canonieke pseudodistributieve wet tussen de vrije coproduct-voltooiingspseudomonade en verschillende vrije limiet-voltooiingspseudomonaden. 
	Wanneer de klasse van limieten pullbacks omvat, leidt deze interactie tot natuurlijke begrippen van \emph{extensieve categorieën}. 
	We tonen aan dat extensieve categorieën met pullbacks en oneindig lex-extensieve categorieën verschijnen als pseudoalgebra’s voor de pseudomonaden die door deze pseudodistributieve wetten worden geïnduceerd. 
	We introduceren het begrip \emph{dubbel-oneindig lex-extensieve categorieën} en bewijzen dat de vrij gegenereerde zulke categorieën cartesiaans gesloten zijn. 
	Hieruit leiden we verder af dat in vrij gegenereerde oneindig lex-extensieve categorieën de objecten met een eindig aantal verbonden componenten exponentieerbaar zijn.
	
	\Needspace{20\baselineskip}
	\item \textbf{Grothendieck constructions}\\
	(\href{https://arxiv.org/pdf/2405.07724}{arXiv:2405.07724})\\
	Dit hoofdstuk komt overeen met \textit{F.~Lucatelli Nunes} en \textit{M.~Vákár}, 
	``Monoidal Closure of Grothendieck Constructions via $\Sigma$-tractable Monoidal Structures and Dialectica Formulas,'' 
	\textit{Theory and Applications of Categories} \textbf{44}(35), 1153--1217 (2025); zie de bibliografie~\cite{nunes_monoidal_grothendieck_2024} voor de publicatiegegevens. 
	We bestuderen Grothendieck-constructies van geïndexeerde categorieën en introduceren het begrip \emph{left-Kan-extensiviteit}, dat de bestaande kennis over extensieve geïndexeerde categorieën verfijnt en uitbreidt en zo het bredere landschap van extensiviteitsbegrippen verrijkt. 
	Dit begrip biedt een systematisch kader voor het analyseren en berekenen van colimieten binnen Grothendieck-constructies. 
	Om monoidale geslotenheid te onderzoeken, introduceren we bovendien het begrip \emph{$\Sigma$-(co)tractabele monoidale structuren} en tonen we aan dat dit kader voldoende voorwaarden oplevert voor het bestaan van gesloten monoidale structuren op Grothendieck-constructies.
\Needspace{5\baselineskip}
	\item \textbf{Conclusie en toekomstig werk}
\end{enumerate}

\cleardoublepage
\chapter{Summary}
Category theory provides a powerful foundation for programming language semantics, enabling rigorous mathematical modeling of type systems, programs, and program equivalences. 
Classic works such as J.~Lambek and P.~Scott’s \textit{Introduction to Higher-Order Categorical Logic}, R.~Crole’s \textit{Categories for Types}, and B.~Jacobs’ \textit{Categorical Logic and Type Theory} demonstrate how categorical structures can model a wide range of type theories and programming languages. 
Building on these foundations, E.~Moggi’s introduction of monads offered a unifying categorical framework for representing computational effects, including state, exceptions, and nondeterminism. 
Since then, the influence of category theory has only deepened, both in analysing existing programming language features and in inspiring new ones. 
For example, M.~Fiore and G.~Plotkin’s work on axiomatic domain theory identified categorical conditions, such as algebraic compactness, sufficient to solve recursive domain equations, thereby providing a categorical generalization of domain semantics. 
Similarly, G.~Plotkin and J.~Power’s theory of algebraic effects introduced a modular approach to handling effects through algebraic theories, refining and extending Moggi’s monadic perspective. 
More recently, P.~B.~Levy’s Call-by-Push-Value unified call-by-name and call-by-value evaluation strategies within a single calculus, supported by an elegant categorical semantics based on an adjunction decomposition of Moggi’s framework. 
Together, these developments demonstrate how category theory continues to provide a robust and unifying foundation for the semantics of programming languages.\par\medskip
\textbf{This thesis: } Building on this tradition, this thesis advances programming language semantics through the lens of categorical structures, guided at every step by a concrete practical problem: \emph{structure-preserving program transformations for computing derivatives, a technique commonly known as Automatic Differentiation (AD)}.
Our theory is developed hand-in-hand with implementation challenges, 
ensuring our categorical constructions yield  realistic program transformations. These transformations are compositional, with correctness established by categorical logical-relations proofs, and provide a basis for further analysis of their complexity.
\par\medskip
\noindent
Our central approach is to view the syntax of a programming language, modulo its standard $\beta\eta$-equational theory, as a suitable \emph{freely generated} categorical structure. Program transformations that preserve program structure then correspond to \emph{structure-preserving (or homomorphic) functors} induced by the universal properties of these free constructions. Using this principle, we define \emph{forward-mode AD}, in the style of dual numbers, as a structure-preserving functor from the syntactic category of a recursive higher-order source language to that of an extended target language, sending each primitive operation to code for its value and derivative action. We prove correctness using a novel \emph{logical-relations technique} for languages featuring recursive types and partial primitives.\par\medskip
\noindent
We then extend this homomorphic perspective on AD to the more challenging and popular \emph{reverse-mode AD} method, using \textbf{CHAD} (Combinatory Homomorphic AD), which handles higher-order \emph{total} languages with expressive features: products, sums, inductive and coinductive types, and function types.
CHAD's semantics is developed using the total categories (Grothendieck constructions) associated with \emph{indexed categories}, and it introduces the notion of \emph{extensive indexed categories} to handle disjoint sums. Correctness follows from the universal property and an appropriate categorical logical-relations technique.\par\medskip
\noindent
We further broaden CHAD to \emph{partial} first-order languages with general iteration. Iteration-extensive indexed categories relate loop substitution to parameterised initial algebras; iterative folders supply the operations needed by the target language. This allows us to interpret unbounded loops and prove the correctness of the extended transformation.\par\medskip
\noindent
In the final chapters, we show that much of the mathematical machinery developed to understand CHAD extends naturally to yield foundational results in fibred category theory, clarifying the structure of Grothendieck constructions of indexed categories, as well as in the theory of distributive and extensive categories.\\

The full thesis is structured as follows:
\begin{enumerate}[leftmargin=2em,itemsep=6pt,parsep=0pt]
	\Needspace{5\baselineskip}
	\item \textbf{Introduction}
	\Needspace{5\baselineskip}
	\item \textbf{Automatic Differentiation for ML-family Languages}\\
	(\href{https://arxiv.org/pdf/2210.07724}{arXiv:2210.07724})\\
	This chapter corresponds to \textit{F.~Lucatelli Nunes} and \textit{M.~Vákár}, 
	``Automatic Differentiation for ML-family languages: correctness via logical relations,''
	\textit{Mathematical Structures in Computer Science} \textbf{34}(8), 747--806 (2024), 
	also available as \href{https://arxiv.org/abs/2210.07724}{arXiv:2210.07724}.
	We develop forward- and reverse-mode automatic differentiation (AD) via dual numbers for 
	ML-family languages featuring recursive types and establish its correctness using 
	logical relations. The AD macro is formulated as a structure-preserving morphism between the syntactic models of the source and target languages, both with recursive types. 
	To prove correctness, we introduce a streamlined and reusable logical-relations framework 
	for reasoning about recursive types and languages with effects.
	
	\Needspace{5\baselineskip}
	\item \textbf{CHAD for Expressive Total Languages}\\
	(\href{https://arxiv.org/pdf/2110.00446}{arXiv:2110.00446})\\
	This chapter corresponds to \textit{F.~Lucatelli Nunes} and \textit{M.~Vákár}, 
	``CHAD for expressive total languages,'' 
	\textit{Mathematical Structures in Computer Science} \textbf{33}(4--5), 311--426 (2023), 
	\href{https://doi.org/10.1017/S096012952300018X}{doi:10.1017/S096012952300018X}, 
	also available as \href{https://arxiv.org/abs/2110.00446}{arXiv:2110.00446}.
	We develop CHAD for total languages with sums, products, inductive, coinductive, and function types. 
	To achieve this, we integrate inductive and coinductive types into dependent type semantics and clarify their intricate relationship through the novel notion of \emph{extensive indexed categories}. 
	We further introduce powerful categorical logical-relations techniques for reasoning about programming languages, drawing on results concerning monadicity and comonadicity.
	
	\Needspace{5\baselineskip}
	\item \textbf{Unraveling the Iterative CHAD}\\
	(\href{https://arxiv.org/pdf/2505.15002}{arXiv:2505.15002})\\
	This chapter corresponds to \textit{F.~Lucatelli Nunes}, \textit{G.~Plotkin} and \textit{M.~Vákár}, 
	``Unraveling the Iterative CHAD,'' 
	\textit{Preprint}, \href{https://arxiv.org/abs/2505.15002v3}{arXiv:2505.15002v3} (2025), revised 13 September 2026.
	We extend CHAD to first-order languages featuring iteration and non-termination, integrating fixpoints through the notion of \emph{iteration-extensive indexed categories}, introduced therein. 
	To this end, we develop several novel notions of \emph{extensive indexed categories} and, more significantly, the concept of \emph{iteration-extensive indexed categories}, which provide the categorical foundations for supporting certain colimits and iteration in the corresponding Grothendieck constructions. 
	The corresponding \emph{iterative folders} specify the operations needed to express the transformation in the target language.

	\Needspace{5\baselineskip}
	\item \textbf{Free Doubly-Infinitary Distributive Categories}\\ (\href{https://arxiv.org/pdf/2403.10447}{arXiv:2403.10447})\\
	This chapter corresponds to \textit{F.~Lucatelli Nunes} and \textit{M.~Vákár}, 
	``Free Doubly-Infinitary Distributive Categories are Cartesian Closed,'' 
	\textit{Applied Categorical Structures} \textbf{34}, article~53 (2026)~\cite{nunes_vakar_2024_doubly_inf_distrib}, also available as \href{https://arxiv.org/abs/2403.10447}{arXiv:2403.10447}.
	Among other contributions, we establish the concept of \emph{doubly-infinitary distributive categories}, that is, categories in which arbitrary small products distribute over arbitrary small coproducts in the precise sense developed in Chapter~5. 
	We further prove that free doubly-infinitary distributive categories are cartesian closed.

	\Needspace{5\baselineskip}
	\item \textbf{Free Extensivity via Distributivity}\\
	(\href{https://arxiv.org/pdf/2405.02185}{arXiv:2405.02185})\\
	This chapter corresponds to \textit{F.~Lucatelli Nunes}, \textit{R.~Prezado}, and \textit{M.~Vákár}, 
	``Free Extensivity via Distributivity,'' 
	\textit{Portugaliae Mathematica} \textbf{82}(1--2), 177--204 (2025), 
	also available as \href{https://arxiv.org/abs/2405.02185}{arXiv:2405.02185}.
	We study the canonical pseudodistributive law between the free coproduct completion pseudomonad and various free limit completion pseudomonads. 
	When the class of limits includes pullbacks, this interaction gives rise to natural notions of \emph{extensive categories}. 
	In particular, we show that extensive categories with pullbacks and infinitary lextensive categories arise as the pseudoalgebras for the pseudomonads induced by these pseudodistributive laws. 
	We introduce the notion of \emph{doubly-infinitary lextensive categories} and prove that the freely generated such categories are cartesian closed. 
	From this, we further deduce that, in freely generated infinitary lextensive categories, the objects with finitely many connected components are exponentiable.
	
	\Needspace{20\baselineskip}
	\item \textbf{Grothendieck constructions}\\ (\href{https://arxiv.org/pdf/2405.07724}{arXiv:2405.07724})\\
	This chapter corresponds to \textit{F.~Lucatelli Nunes} and \textit{M.~Vákár}, 
	``Monoidal Closure of Grothendieck Constructions via $\Sigma$-tractable Monoidal Structures and Dialectica Formulas,'' published in 
	\textit{Theory and Applications of Categories} \textbf{44}(35), 1153--1217 (2025)~\cite{nunes_monoidal_grothendieck_2024}, 
	\href{https://arxiv.org/abs/2405.07724}{arXiv:2405.07724} [math.CT].
	We study Grothendieck constructions of indexed categories, introducing the notion of \emph{left-Kan extensivity}, which refines and extends the existing understanding of extensive indexed categories while enriching the broader landscape of extensivity-like notions. 
	This concept provides a systematic framework for analyzing and computing colimits within Grothendieck constructions. 
	
	Furthermore, to investigate monoidal closedness, we introduce the notion of \emph{$\Sigma$-(co)tractable monoidal structures} and show that this framework yields sufficient conditions for the existence of closed monoidal structures on Grothendieck constructions.

	\Needspace{5\baselineskip}
	\item \textbf{Conclusion and Future Work}
\end{enumerate}


\addcontentsline{toc}{chapter}{Bibliography}
\begin{thebibliography}{40}
\providecommand{\natexlab}[1]{#1}
\providecommand{\url}[1]{\texttt{#1}}
\expandafter\ifx\csname urlstyle\endcsname\relax
  \providecommand{\doi}[1]{doi: #1}\else
  \providecommand{\doi}{doi: \begingroup \urlstyle{rm}\Url}\fi

\bibitem[Abadi and Plotkin(2020)]{PLOTKU}
Mart{\'i}n Abadi and Gordon~D. Plotkin.
\newblock {{A Simple Differentiable Programming Language}}.
\newblock \emph{Proceedings of the ACM on Programming Languages}, 4\penalty0
  (POPL):\penalty0 38:1--38:28, January 2020.
\newblock \doi{10.1145/3371106}.
\newblock URL \url{https://doi.org/10.1145/3371106}.

\bibitem[Abadi et~al.(2015)Abadi, Agarwal, Barham, Brevdo, Chen, Citro,
  Corrado, Davis, Dean, Devin, Ghemawat, Goodfellow, Harp, Irving, Isard, Jia,
  Jozefowicz, Kaiser, Kudlur, Levenberg, Man{\'e}, Monga, Moore, Murray, Olah,
  Schuster, Shlens, Steiner, Sutskever, Talwar, Tucker, Vanhoucke, Vasudevan,
  Vi{\'e}gas, Vinyals, Warden, Wattenberg, Wicke, Yu, and
  Zheng]{tensorflow2015}
Mart{\'i}n Abadi, Ashish Agarwal, Paul Barham, Eugene Brevdo, Zhifeng Chen,
  Craig Citro, Greg~S. Corrado, Andy Davis, Jeffrey Dean, Matthieu Devin,
  Sanjay Ghemawat, Ian Goodfellow, Andrew Harp, Geoffrey Irving, Michael Isard,
  Yangqing Jia, Rafal Jozefowicz, Lukasz Kaiser, Manjunath Kudlur, Josh
  Levenberg, Dandelion Man{\'e}, Rajat Monga, Sherry Moore, Derek Murray, Chris
  Olah, Mike Schuster, Jonathon Shlens, Benoit Steiner, Ilya Sutskever, Kunal
  Talwar, Paul Tucker, Vincent Vanhoucke, Vijay Vasudevan, Fernanda Vi{\'e}gas,
  Oriol Vinyals, Pete Warden, Martin Wattenberg, Martin Wicke, Yuan Yu, and
  Xiaoqiang Zheng.
\newblock {{TensorFlow: Large-Scale Machine Learning on Heterogeneous
  Systems}}, 2015.
\newblock URL \url{https://www.tensorflow.org/}.
\newblock Software.

\bibitem[Abbott et~al.(2003)Abbott, Altenkirch, and
  Ghani]{AbbottAltenkirchGhaniFoSSaCS03}
Michael Abbott, Thorsten Altenkirch, and Neil Ghani.
\newblock {{Categories of Containers}}.
\newblock In \emph{FoSSaCS 2003}, volume 2620 of \emph{Lecture Notes in
  Computer Science}, pages 23--38. Springer, 2003.
\newblock \doi{10.1007/3-540-36576-1_2}.
\newblock URL \url{https://people.cs.nott.ac.uk/psztxa/publ/fossacs03.pdf}.

\bibitem[Ad{\'a}mek and Rosick{\'y}(2020)]{AdamekFreeCoproductCompletion}
Ji{\v r}{\'i} Ad{\'a}mek and Ji{\v r}{\'i} Rosick{\'y}.
\newblock {{How nice are free completions of categories?}}
\newblock \emph{Topology and its Applications}, 273:\penalty0 106972, 2020.
\newblock \doi{10.1016/j.topol.2019.106972}.
\newblock URL
  \url{https://www.sciencedirect.com/science/article/pii/S0166864119303773}.

\bibitem[Ad\'{a}mek et~al.(2006)Ad\'{a}mek, Herrlich, and
  Strecker]{AdamekHerrlichStrecker2006}
Ji\v{r}\'{i} Ad\'{a}mek, Horst Herrlich, and George~E. Strecker.
\newblock \emph{{{Abstract and Concrete Categories: The Joy of Cats}}}.
\newblock Theory and Applications of Categories, 2006.
\newblock URL \url{https://www.tac.mta.ca/tac/reprints/articles/17/tr17.pdf}.
\newblock Reprint, available online.

\bibitem[Baydin et~al.(2018)Baydin, Pearlmutter, Radul, and
  Siskind]{baydin2018automatic}
Atilim~Gunes Baydin, Barak~A. Pearlmutter, Alexey~Andreyevich Radul, and
  Jeffrey~Mark Siskind.
\newblock {{Automatic Differentiation in Machine Learning: A Survey}}.
\newblock \emph{Journal of Machine Learning Research}, 18\penalty0
  (153):\penalty0 1--43, 2018.
\newblock URL \url{http://jmlr.org/papers/v18/17-468.html}.

\bibitem[Bradbury et~al.(2018)Bradbury, Frostig, Hawkins, Johnson, Katariya,
  Leary, Maclaurin, Necula, Paszke, Vander{P}las, Wanderman-{M}ilne, and
  Zhang]{jax2018github}
James Bradbury, Roy Frostig, Peter Hawkins, Matthew~James Johnson, Yash
  Katariya, Chris Leary, Dougal Maclaurin, George Necula, Adam Paszke, Jake
  Vander{P}las, Skye Wanderman-{M}ilne, and Qiao Zhang.
\newblock {{JAX: composable transformations of Python+NumPy programs}}, 2018.
\newblock URL \url{https://github.com/jax-ml/jax}.
\newblock Software.

\bibitem[Crole(1993)]{crole1993categories}
Roy~L. Crole.
\newblock \emph{{{Categories for Types}}}.
\newblock Cambridge Mathematical Textbooks. Cambridge University Press, 1993.
\newblock \doi{10.1017/CBO9781139172707}.
\newblock URL
  \url{http://www.cambridge.org/de/academic/subjects/computer-science/programming-languages-and-applied-logic/categories-types?format=PB}.

\bibitem[de~Amorim(2022)]{deamorim2022pittsrelationalpropertiesdomains}
Arthur~Azevedo de~Amorim.
\newblock {{On Pitts' Relational Properties of Domains}}.
\newblock arXiv:2207.07053, 2022.
\newblock URL \url{https://arxiv.org/abs/2207.07053}.

\bibitem[Elliott(2018)]{elliott2018essence}
Conal Elliott.
\newblock {{The simple essence of automatic differentiation}}.
\newblock \emph{Proc. ACM Program. Lang.}, 2\penalty0 (ICFP):\penalty0
  70:1--70:29, jul 2018.
\newblock \doi{10.1145/3236765}.
\newblock URL \url{https://doi.org/10.1145/3236765}.

\bibitem[Fiore and Plotkin(1997)]{fiore1996axiomatic}
Marcelo~P. Fiore and Gordon~D. Plotkin.
\newblock {{An extension of models of Axiomatic Domain Theory to models of
  Synthetic Domain Theory}}.
\newblock In Dirk van Dalen and Marc Bezem, editors, \emph{Computer Science
  Logic}, volume 1258 of \emph{Lecture Notes in Computer Science}, pages
  129--149, Berlin, Heidelberg, 1997. Springer Berlin Heidelberg.
\newblock \doi{10.1007/3-540-63172-0_36}.

\bibitem[Griewank and Walther(2008)]{griewank2008evaluating}
Andreas Griewank and Andrea Walther.
\newblock \emph{{{Evaluating Derivatives: Principles and Techniques of
  Algorithmic Differentiation}}}.
\newblock Society for Industrial and Applied Mathematics, second edition, 2008.
\newblock \doi{10.1137/1.9780898717761}.
\newblock URL \url{https://epubs.siam.org/doi/abs/10.1137/1.9780898717761}.

\bibitem[Hoofman and Moerdijk(1995)]{hoofmanMoerdijk1995semifunctors}
Raymond Hoofman and Ieke Moerdijk.
\newblock A remark on the theory of semi-functors.
\newblock \emph{Mathematical Structures in Computer Science}, 5\penalty0
  (1):\penalty0 1--8, 1995.
\newblock \doi{10.1017/S096012950000061X}.

\bibitem[Huot et~al.(2020)Huot, Staton, and
  V{\'a}k{\'a}r]{DBLP:conf/fossacs/HuotSV20}
Mathieu Huot, Sam Staton, and Matthijs V{\'a}k{\'a}r.
\newblock {{Correctness of Automatic Differentiation via Diffeologies and
  Categorical Gluing}}.
\newblock In Jean Goubault-Larrecq and Barbara K{\"o}nig, editors,
  \emph{Foundations of Software Science and Computation Structures}, pages
  319--338, Cham, 2020. Springer International Publishing.
\newblock \doi{10.1007/978-3-030-45231-5_17}.

\bibitem[Jacobs(1999)]{jacobs1999categorical}
Bart Jacobs.
\newblock \emph{{{Categorical Logic and Type Theory}}}, volume 141 of
  \emph{Studies in Logic and the Foundations of Mathematics}.
\newblock Elsevier, 1999.

\bibitem[Joss(1976)]{joss1976algorithmisches}
Johann Joss.
\newblock \emph{{{Algorithmisches Differenzieren}}}.
\newblock PhD thesis, ETH Z{\"u}rich, 1976.
\newblock URL
  \url{https://www.research-collection.ethz.ch/handle/20.500.11850/134597}.
\newblock Diss.\ Math.\ ETH Z{\"u}rich, Nr.\ 5757.

\bibitem[Knuth(1991)]{knuth1991theorypractice}
Donald~E. Knuth.
\newblock {{Theory and practice}}.
\newblock \emph{Theoretical Computer Science}, 90:\penalty0 1--15, 1991.
\newblock URL \url{https://arxiv.org/abs/cs/9301114}.

\bibitem[Lambek and Scott(1986)]{lambek1986introduction}
Joachim Lambek and Philip~J. Scott.
\newblock \emph{{{Introduction to Higher Order Categorical Logic}}}, volume~7
  of \emph{Cambridge Studies in Advanced Mathematics}.
\newblock Cambridge University Press, 1986.

\bibitem[Levy(1999)]{levy1999call}
Paul~Blain Levy.
\newblock {{Call-by-Push-Value: A Subsuming Paradigm}}.
\newblock In \emph{Typed Lambda Calculi and Applications}, volume 1581 of
  \emph{Lecture Notes in Computer Science}, pages 228--243. Springer, 1999.
\newblock \doi{10.1007/3-540-48959-2_17}.

\bibitem[Levy(2006)]{levy2006cbpv}
Paul~Blain Levy.
\newblock {{Call-by-push-value: Decomposing call-by-value and call-by-name}}.
\newblock \emph{Higher-Order and Symbolic Computation}, 19\penalty0
  (4):\penalty0 377--414, 2006.
\newblock \doi{10.1007/s10990-006-0480-6}.
\newblock URL \url{https://doi.org/10.1007/s10990-006-0480-6}.

\bibitem[{Lucatelli Nunes} and V{\'a}k{\'a}r(2023)]{nunes_vakar_2021_chad}
F.~{Lucatelli Nunes} and Matthijs V{\'a}k{\'a}r.
\newblock {{CHAD for expressive total languages}}.
\newblock \emph{Mathematical Structures in Computer Science}, 33\penalty0
  (4--5):\penalty0 311--426, 2023.
\newblock \doi{10.1017/S096012952300018X}.

\bibitem[{Lucatelli Nunes} and V{\'{a}}k{\'{a}}r(2023)]{nunes_vakar_mfo_4046}
F.~{Lucatelli Nunes} and Matthijs V{\'{a}}k{\'{a}}r.
\newblock {{Logical Relations for Partial Features and Automatic
  Differentiation Correctness}}.
\newblock Technical Report 2023--09, Mathematisches Forschungsinstitut
  Oberwolfach, 2023.
\newblock URL
  \url{https://publications.mfo.de/bitstream/handle/mfo/4046/OWP-2023-09.pdf}.

\bibitem[{Lucatelli Nunes} and V{\'a}k{\'a}r(2024)]{nunes_vakar_2024_ml}
F.~{Lucatelli Nunes} and Matthijs V{\'a}k{\'a}r.
\newblock {{Automatic differentiation for ML-family languages: Correctness via
  logical relations}}.
\newblock \emph{Mathematical Structures in Computer Science}, 34\penalty0
  (8):\penalty0 747--806, 2024.
\newblock \doi{10.1017/S0960129524000215}.

\bibitem[{Lucatelli Nunes} and
  V{\'a}k{\'a}r(2025)]{nunes_monoidal_grothendieck_2024}
F.~{Lucatelli Nunes} and Matthijs V{\'a}k{\'a}r.
\newblock {{Monoidal closure of Grothendieck constructions via
  $\Sigma$-tractable monoidal structures and Dialectica formulas}}.
\newblock \emph{Theory and Applications of Categories}, 44\penalty0
  (35):\penalty0 1153--1217, 2025.
\newblock URL \url{https://www.tac.mta.ca/tac/volumes/44/35/44-35abs.html}.

\bibitem[{Lucatelli Nunes} and
  V{\'a}k{\'a}r(2026)]{nunes_vakar_2024_doubly_inf_distrib}
F.~{Lucatelli Nunes} and Matthijs V{\'a}k{\'a}r.
\newblock {{Free Doubly-Infinitary Distributive Categories are Cartesian
  Closed}}.
\newblock \emph{Applied Categorical Structures}, 34\penalty0 (5):\penalty0 53,
  2026.
\newblock \doi{10.1007/s10485-026-09887-7}.
\newblock URL \url{https://arxiv.org/abs/2403.10447}.

\bibitem[{Lucatelli Nunes} et~al.(2025{\natexlab{a}}){Lucatelli Nunes},
  Plotkin, and V{\'a}k{\'a}r]{nunes_iterative_chad_2025}
F.~{Lucatelli Nunes}, Gordon Plotkin, and Matthijs V{\'a}k{\'a}r.
\newblock {{Unraveling the iterative CHAD}}.
\newblock arXiv:2505.15002, 2025{\natexlab{a}}.
\newblock URL \url{https://arxiv.org/abs/2505.15002v3}.
\newblock Version 3, revised 13 September 2026.

\bibitem[{Lucatelli Nunes} et~al.(2025{\natexlab{b}}){Lucatelli Nunes},
  Prezado, and V{\'a}k{\'a}r]{nunes_extensivity_2025}
F.~{Lucatelli Nunes}, Rui Prezado, and Matthijs V{\'a}k{\'a}r.
\newblock {{Free extensivity via distributivity}}.
\newblock \emph{Portugaliae Mathematica}, 82\penalty0 (1--2):\penalty0
  177--204, 2025{\natexlab{b}}.
\newblock \doi{10.4171/PM/2129}.

\bibitem[{Lucatelli Nunes} et~al.(2026){Lucatelli Nunes}, Simm, and
  V{\'a}k{\'a}r]{nunes_simm_vakar_2026_variants}
F.~{Lucatelli Nunes}, Diogo Simm, and Matthijs V{\'a}k{\'a}r.
\newblock {Simply Typed Reverse-Mode Automatic Differentiation with Variants:
  Denotational Correctness via Idempotent Completion}.
\newblock arXiv:2607.14453, 2026.
\newblock URL \url{https://arxiv.org/abs/2607.14453v3}.

\bibitem[Mac~Lane(1998)]{MacLane1998}
Saunders Mac~Lane.
\newblock \emph{{{Categories for the Working Mathematician}}}, volume~5 of
  \emph{Graduate Texts in Mathematics}.
\newblock Springer, 2 edition, 1998.
\newblock \doi{10.1007/978-1-4757-4721-8}.

\bibitem[Mazza and Pagani(2021)]{DBLP:journals/pacmpl/MazzaP21}
Damiano Mazza and Michele Pagani.
\newblock {{Automatic differentiation in {PCF}}}.
\newblock \emph{Proc. {ACM} Program. Lang.}, 5\penalty0 ({POPL}):\penalty0
  1--27, 2021.
\newblock \doi{10.1145/3434309}.
\newblock URL \url{https://doi.org/10.1145/3434309}.

\bibitem[Moggi(1989)]{moggi1989computational}
E.~Moggi.
\newblock {{Computational lambda-calculus and monads}}.
\newblock In \emph{[1989] Proceedings. Fourth Annual Symposium on Logic in
  Computer Science}, pages 14--23, 1989.
\newblock \doi{10.1109/LICS.1989.39155}.
\newblock URL \url{https://doi.org/10.1109/LICS.1989.39155}.

\bibitem[Moggi(1991)]{moggi1991notions}
Eugenio Moggi.
\newblock {{Notions of computation and monads}}.
\newblock \emph{Information and Computation}, 93\penalty0 (1):\penalty0 55--92,
  1991.
\newblock \doi{10.1016/0890-5401(91)90052-4}.
\newblock URL
  \url{https://www.sciencedirect.com/science/article/pii/0890540191900524}.
\newblock Selections from 1989 IEEE Symposium on Logic in Computer Science.

\bibitem[Paszke et~al.(2019)Paszke, Gross, Massa, Lerer, Bradbury, Chanan,
  Killeen, Lin, Gimelshein, Antiga, et~al.]{paszke2019pytorch}
Adam Paszke, Sam Gross, Francisco Massa, Adam Lerer, James Bradbury, Gregory
  Chanan, Trevor Killeen, Zeming Lin, Natalia Gimelshein, Luca Antiga, et~al.
\newblock {{PyTorch: An Imperative Style, High-Performance Deep Learning
  Library}}.
\newblock In \emph{Advances in Neural Information Processing Systems 32
  (NeurIPS 2019)}, 2019.
\newblock URL
  \url{https://papers.nips.cc/paper/9015-pytorch-an-imperative-style-high-performance-deep-learning-library.pdf}.

\bibitem[Pitts(1996)]{DBLP:journals/iandc/Pitts96}
Andrew~M. Pitts.
\newblock {{Relational Properties of Domains}}.
\newblock \emph{Information and Computation}, 127\penalty0 (2):\penalty0
  66--90, jun 1996.
\newblock \doi{10.1006/inco.1996.0052}.
\newblock URL \url{https://www.cl.cam.ac.uk/~amp12/papers/relpod/relpod.pdf}.

\bibitem[Reynolds(1972)]{reynolds1972definitional}
John~C. Reynolds.
\newblock {{Definitional Interpreters for Higher-Order Programming Languages}}.
\newblock In \emph{Proceedings of the ACM Annual Conference (ACM '72), Volume
  2}, pages 717--740, 1972.
\newblock \doi{10.1145/800194.805852}.

\bibitem[Simm et~al.(2026)Simm, {Lucatelli Nunes}, and
  V{\'a}k{\'a}r]{simm_nunes_vakar_2026_effectful}
Diogo Simm, F.~{Lucatelli Nunes}, and Matthijs V{\'a}k{\'a}r.
\newblock {Backpropagation for Effectful Languages {I}: Finite Probability and
  Discrete Output Algebraic Effects}.
\newblock arXiv:2607.13935, 2026.
\newblock URL \url{https://arxiv.org/abs/2607.13935}.

\bibitem[Smeding and V{\'a}k{\'a}r(2023)]{DBLP:conf/popl/SmedingV23}
Tom Smeding and Matthijs V{\'a}k{\'a}r.
\newblock {{Efficient Dual-Numbers Reverse AD via Well-Known Program
  Transformations}}.
\newblock \emph{Proceedings of the ACM on Programming Languages}, 7\penalty0
  (POPL):\penalty0 1573--1600, 2023.
\newblock \doi{10.1145/3571247}.
\newblock URL \url{https://doi.org/10.1145/3571247}.

\bibitem[Smeding and V{\'a}k{\'a}r(2024)]{DBLP:conf/popl/SmedingV24}
Tom~J. Smeding and Matthijs I.~L. V{\'a}k{\'a}r.
\newblock {{Efficient CHAD}}.
\newblock \emph{Proceedings of the ACM on Programming Languages}, 8\penalty0
  (POPL):\penalty0 1060--1088, 2024.
\newblock \doi{10.1145/3632878}.
\newblock URL \url{https://doi.org/10.1145/3632878}.

\bibitem[Smeding and V{\'a}k{\'a}r(2025)]{smeding2022parallel}
Tom~J. Smeding and Matthijs I.~L. V{\'a}k{\'a}r.
\newblock {{Parallel dual-numbers reverse AD}}.
\newblock \emph{Journal of Functional Programming}, 35:\penalty0 e16, 2025.
\newblock \doi{10.1017/S0956796825100051}.
\newblock URL \url{https://arxiv.org/abs/2207.03418}.

\bibitem[V{\'a}k{\'a}r(2021)]{DBLP:conf/esop/Vakar21}
Matthijs V{\'a}k{\'a}r.
\newblock {{Reverse AD at Higher Types: Pure, Principled and Denotationally
  Correct}}.
\newblock In Nobuko Yoshida, editor, \emph{Programming Languages and Systems},
  pages 607--634, Cham, 2021. Springer International Publishing.
\newblock \doi{10.1007/978-3-030-72019-3_22}.
\newblock URL \url{https://doi.org/10.1007/978-3-030-72019-3\_22}.

\end{thebibliography}
\end{document}